\documentclass[a4paper,10pt]{article}
\usepackage[utf8x]{inputenc}

\usepackage{amsmath,amssymb,amsfonts,amsthm,dsfont,textgreek}

\usepackage{bbm}
\usepackage{hyperref}

\title{Lieb-Thirring and Cwickel-Lieb-Rozenblum inequalities for perturbed graphene with a Coulomb impurity}
\author{Sergey Morozov and David M\"uller}
\date{}

\newtheorem{thm}{Theorem}
\newtheorem{lem}[thm]{Lemma}
\newtheorem{rem}[thm]{Remark}
\newtheorem{cor}[thm]{Corollary}
\newtheorem{definition}[thm]{Definition}

\let\Im\undefined
\DeclareMathOperator{\Im}{Im}
\let\Re\undefined
\DeclareMathOperator{\Re}{Re}

\DeclareMathOperator{\tr}{tr}
\DeclareMathOperator{\dd}{d\!}
\DeclareMathOperator{\rank}{rank}

\DeclareMathOperator{\loc}{loc}
\DeclareMathOperator{\Span}{span}
\DeclareMathOperator{\diag}{diag}
\DeclareMathOperator{\supp}{supp}
\DeclareMathOperator{\Ltwolim}{\mathsf L^2-lim}

\begin{document}

\maketitle

\begin{center} Mathematisches Institut, Ludwig-Maximilians-Universit\"at M\"unchen\\ Theresienstr. 39, 80333 Munich, Germany.\\ \smallskip $\mathtt{ \underline{morozov@math.lmu.de},\ dmueller@math.lmu.de}$\end{center}

\medskip
\noindent{\bfseries Abstract:} We study the two dimensional massless Coulomb-Dirac operator restricted to its positive spectral subspace and prove estimates on the negative eigenvalues created by electromagnetic perturbations.



\section{Introduction}

As theoretically predicted \cite{SilvestrovEfetov2007,Downing2011} and recently experimentally observed \cite{LeeWong2016,GutierrezBrown2016}, electrostatic potentials can create bound states in graphene, which corresponds to emergence of a quantum dot. We consider the case of a graphene sheet with an attractive Coulomb impurity perturbed by a weak electromagnetic potential and provide bounds on the energies of the bound states using a model similar to those of \cite{EggerSiedentop2010}.

We begin by considering a gra\-phene sheet with an attractive Coulomb impurity of strength $\nu$. For energies near the conical point of the energy-quasi-momentum dispersion relation, the Hamiltonian of an electron in such material is effectively given by
the massless Coulomb-Dirac operator (see \cite{Pereira-Nilsson-Castro_Neto} and Section IV of \cite{CastroNeto_et_al}). This operator acts in $\mathsf L^2(\mathbb R^2, \mathbb C^2)$ and is associated to the differential expression
\begin{equation}\label{Dirac symbol}
 d^\nu:= -\mathrm i\boldsymbol\sigma\cdot\nabla- \nu|\cdot|^{-1}.
\end{equation}
Here the units are chosen so that the Fermi velocity $v_F$ equals $1$, and $\boldsymbol\sigma= (\sigma_1, \sigma_2)= \bigg(\begin{pmatrix}0&1\\1&0\end{pmatrix}, \begin{pmatrix}0&-\mathrm i\\ \mathrm i&0\end{pmatrix}\bigg)$ is the vector of Pauli matrices.
For $\nu\in[0,1/2]$ (which we assume throughout in the following) we work with the distinguished self-adjoint operator $D^\nu$ in $\mathsf L^2(\mathbb R^2,\mathbb C^2)$ associated to \eqref{Dirac symbol} (see \cite{Mueller, Warmt} and \eqref{D^nu} below). The supercritical case of $\nu >1/2$ is not considered here. In that case the canonical choice of a particular self-adjoint realisation among many possible is not well-established.

We now state the main results of the paper. Scalar operators like $\sqrt{-\Delta}$ are applied to vector-valued functions component-wise without reflecting this in the notation.

\begin{thm}~\label{t:abs_value_bound}
\begin{enumerate}
  \item For every $\nu\in [0, 1/2)$ there exists $C_\nu >0$ such that
\begin{align}\label{CLR reduction}
 |D^\nu| \geqslant C_\nu\sqrt{-\Delta}
\end{align}
 holds.
  \item For any $\lambda\in [0, 1)$ there exists $K_\lambda >0$ such that
\begin{align}\label{l bound}
 |D^{1/2}| \geqslant K_\lambda l^{\lambda -1}(-\Delta)^{\lambda/2} -l^{-1}
\end{align}
holds for any $l >0$.
\end{enumerate}
\end{thm}

Note that for $\nu \in (0, 1/2]$ the inequality $(D^\nu)^2 \geqslant C(-\Delta)$ is false for any $C >0$, since by Corollary \ref{c: operator core} the operator domain of $D^\nu$ is not contained in $\mathsf H^1(\mathbb R^2, \mathbb C^2)$.

The operator inequality \eqref{l bound} is related to the estimate for the fractional Schr\"o\-din\-ger operator with Coulomb potential in $\mathsf L^2(\mathbb R^2)$: For any $t\in (0, 1/2)$ there exists $M_t >0$ such that
\begin{align}\label{SSS}
 (-\Delta)^{1/2} -\frac{2\big(\Gamma(3/4)\big)^2}{\big(\Gamma(1/4)\big)^2|\cdot|} \geqslant M_t l^{2t -1}(-\Delta)^t -l^{-1}
\end{align}
holds for all $l >0$, see (1.3) in \cite{FrankHLT} (and Theorem 2.3 in \cite{SolovejSoerensenSpitzer} for an analogous result in three dimensions).

Since the negative energy states of $D^\nu$ belong to the fully occupied valence band of graphene (Dirac sea) \cite{Wallace1947,CastroNeto_et_al}, the space of physically available electronic states is $\mathfrak H_+^\nu:= P^\nu_+\mathsf L^2(\mathbb R^2,\mathbb C^2)$, where $P^\nu_+$ is the spectral projector of $D^\nu$ to the half-line $[0, \infty)$. We now consider perturbations of $D^\nu$ by electromagnetic potentials, which are assumed to be weak enough so that the state space is essentially unchanged:

\begin{cor}\label{defining corollary}
Suppose that $(\nu, \gamma)\in \big([0, 1/2]\times [0, \infty)\big)\setminus \big\{(1/2, 0)\big\}$.
Let $V$ be a non-negative measurable $(2\times 2)$-matrix function with $\tr (V^{2 +\gamma})\in \mathsf L^1(\mathbb R^2)$.
Let $\mathfrak w$ be a real-valued quadratic form in $\mathfrak H_+^\nu$ with the domain containing $P_+^\nu\mathfrak D\big(|D^\nu|^{1/2}\big)$. Assume that there exists $C> 0$ such that
\begin{align}\label{w}
 0 \leqslant\mathfrak w[\varphi] \leqslant C\Big(\big\||D^\nu|^{1/2}\varphi\big\|^2 +\|\varphi\|^2\Big), \quad\textrm{for all }\varphi \in P_+^\nu\mathfrak D\big(|D^\nu|^{1/2}\big).
\end{align}
Then the quadratic form
\begin{align*}
    \mathfrak d^\nu(\mathfrak w, V)&:P_+^\nu\mathfrak D\big(|D^\nu|^{1/2}\big)\to \mathbb R, \\ \mathfrak d^\nu(\mathfrak w, V)[\varphi] &:=\big\||D^\nu|^{1/2}\varphi\big\|^2 +\mathfrak w[\varphi] - \int\limits_{\mathbb{R}^2}\big\langle \varphi(\mathbf x), V(\mathbf x)\varphi (\mathbf x)\big\rangle \mathrm{d}\mathbf{x}
\end{align*}
is closed and bounded from below in $\mathfrak H_+^\nu$.
\end{cor}
According to Theorem 10.1.2 in \cite{BirmanSolomjak}, there exists a unique self-adjoint operator $D^\nu(\mathfrak w, V)$ in $\mathfrak H_+^\nu$ associated to $\mathfrak d^\nu(\mathfrak w, V)$.
In the following two theorems we study the negative spectrum of $D^\nu(\mathfrak w, V)$. Note that the eigenvalues of $D^\nu(\mathfrak w, V)$ can be interpreted as bound states of a quantum dot.

For numbers and self-adjoint operators we use the notation $x_\pm := \max\{\pm x, 0\}$ for the positive and negative parts of $x$.

\begin{thm}\label{CLR theorem}
 Let $\nu\in [0, 1/2)$. There exists $C^{\mathrm{CLR}}_\nu> 0$ such that
\begin{equation}\label{CLR}
 \rank \big(D^\nu(\mathfrak w, V)\big)_- \leqslant C^{\mathrm{CLR}}_\nu\int_{\mathbb R^2} \tr \big(V(\mathbf x)\big)^2\dd\mathbf x.
\end{equation}
\end{thm}

Analogues of Theorem \ref{CLR theorem} are widely known for many bounded from below self-adjoint operators as Cwickel-Lieb-Rozenblum inequalities (see \cite{Rozenblum, Cwikel, Lieb1976} for the original contributions and \cite{Frank2014CLR} and references therein for further developments). In particular, in Example 3.3 of \cite{Frank2014CLR} it is proved that the estimate
\begin{align}\label{CLR for fractional Laplacian}
 \rank\big((-\Delta)^{t} -V\big)_- \leqslant (4\pi t)^{-1}(1 -t)^{(t -2)/t}\int_{\mathbb R^2}\tr \big(V(\mathbf x)\big)^{1/t}\mathrm d\mathbf x
\end{align}
holds for all $0 <t <1$. Our proof of Theorem \ref{CLR theorem} is based on Theorem \ref{t:abs_value_bound} and \eqref{CLR for fractional Laplacian}.

\begin{thm}\label{LT theorem}
 Let $\nu\in [0, 1/2]$ and $\gamma> 0$. There exists $C^{\mathrm{LT}}_{\nu, \gamma}> 0$ such that
\begin{equation}\label{LT}
 \tr\big(D^\nu(\mathfrak w, V)\big)^\gamma_- \leqslant C^{\mathrm{LT}}_{\nu, \gamma}\int_{\mathbb R^2} \tr \big(V(\mathbf x)\big)^{2+ \gamma}\dd\mathbf x.
\end{equation}
\end{thm}

Theorem \ref{LT theorem} is a form of Lieb-Thirring inequality, (see \cite{LiebThirring} for the original result and \cite{LaptevWeidl} for a review of further developments). In another publication \cite{Virtual_level} we prove that $D^{1/2}(0, V)$ has a negative eigenvalue for any non-trivial $V\geqslant 0$. This situation is associated with the existence of a virtual level at zero, as observed for example for the operator $\Big(-\dfrac{\dd^{\,2}}{\dd r^2}- \dfrac1{4r^2}\Big)$ in $\mathsf L^2(\mathbb R_+)$ (see \cite{EkholmFrank}, Proposition 3.2). In particular, the bound \eqref{CLR} cannot hold for $\nu=1/2$. In this case Theorem \ref{LT theorem} is an equivalent of Hardy-Lieb-Thirring inequality (see \cite{EkholmFrank, FrankHLT, FrankLiebSeiringer}).

Certain estimates for the optimal constants in Theorems \ref{t:abs_value_bound} -- \ref{LT theorem} can be extracted from the proofs provided. This results in explicit, but quite involved expressions.

The article is organised in the following way: We start with some auxiliary results in Section \ref{transforms section}, where we prepare useful representations of operators of interest with the help of certain unitary transforms. One of such representations allows us to provide a rigorous definition \eqref{D^nu} of $D^{\nu}$. In Section \ref{scalar section} we study the operator $(-\Delta)^{1/2}-\alpha|\cdot|^{-1}$ in the representation, in which it can be relatively easily compared with $|D^\nu|$. Such comparison is done in the two critical channels of the angular momentum decomposition in Section \ref{critical channels section}. For the non-critical channels we obtain a lower bound on $|D^\nu|$ in terms of $(-\Delta)^{1/2}$ in Section \ref{non-critical section}. In the subsequent Section \ref{clb section} we prove a channel-wise improvement of \eqref{SSS}. Finally, in Section \ref{main proofs section} we complete the proofs of Theorems \ref{t:abs_value_bound} -- \ref{LT theorem} and Corollary 
\ref{defining corollary}.

\paragraph{Acknowledgement:} S. M. was supported by the RSF grant 15-11-30007.

\section{Mellin, Fourier and related transforms in polar coordinates}\label{transforms section}

In this section we introduce several unitary transformations which will be useful in the subsequent analysis. We also formulate and prove several technical results needed in the subsequent sections.
Let $(r, \theta)$, $(p, \omega)\in [0, \infty)\times[0, 2\pi)$ be the polar coordinates in $\mathbb R^2$ in coordinate and momentum spaces, respectively.

\paragraph{Fourier transform.}
We use the standard unitary Fourier transform in $\mathsf L^2(\mathbb R^2)$ given in the polar coordinates for $\varphi\in \mathsf L^1(\mathbb R^2)\cap \mathsf L^2(\mathbb R^2)$ by
\begin{equation}\label{Fourier transform}
 (\mathcal F\varphi)(p, \omega):= \frac1{2\pi}\int_0^\infty\int_0^{2\pi}\mathrm e^{-\mathrm i pr\cos(\omega- \theta)}\varphi(r, \theta)\, \mathrm d\theta\,r\mathrm dr.
\end{equation}

\begin{lem}\label{Fourier channel lemma}
 For $m\in \mathbb Z$ and $\psi\in \mathsf C_0^\infty\big([0, \infty)\big)$ the Fourier transform of
\begin{equation}\label{auxiliary nonsense}
 \Psi^{(m)}(r, \theta): = r^{-1/2}\psi(r)\mathrm e^{\mathrm im\theta}
\end{equation}
 is given in the polar coordinates by
\begin{equation}\label{Fourier channel}
 \mathcal F(\Psi^{(m)})(p, \omega)= (-\mathrm i)^m\mathrm e^{\mathrm im\omega}\int_0^\infty \sqrt{r}J_m(pr)\psi(r)\mathrm dr.
\end{equation}
\end{lem}

\begin{proof}
 According to \cite{dlmf}, 10.9.2 and 10.2.2
\begin{equation}\label{aux1}
 \int_0^{2\pi}\mathrm e^{-\mathrm ipr\cos(\omega- \theta)}\mathrm e^{\mathrm im\theta}\mathrm d\theta= 2\pi\mathrm i^m J_m(-pr)\mathrm e^{\mathrm im\omega}= 2\pi(-\mathrm i)^mJ_m(pr)\mathrm e^{\mathrm im\omega}.
\end{equation}
Substituting \eqref{auxiliary nonsense} into \eqref{Fourier transform} and using \eqref{aux1} we obtain \eqref{Fourier channel}.
\end{proof}

\paragraph{Mellin transform.} Let $\mathcal M$ be the unitary Mellin transform, first defined on $\mathsf C_0^\infty(\mathbb R_+)$ by
\begin{equation}\label{Mellin transform}
 (\mathcal M\psi)(s):= \frac1{\sqrt{2\pi}}\int_0^\infty r^{-1/2- \mathrm is}\psi(r)\mathrm dr,
\end{equation}
and then extended to a unitary operator $\mathcal M:\mathsf L^2(\mathbb R_+)\to \mathsf L^2(\mathbb R)$, see e.g. \cite{YaouancOliverRaynal}.

\begin{definition}\label{D lambda definition}
For $\lambda\in \mathbb R\setminus\{0\}$ let $\mathfrak D^\lambda$ be the set of functions $\psi\in \mathsf L^2(\mathbb R)$ such that there exists $\Psi$ analytic in the strip $\mathfrak S^\lambda:= \big\{z\in\mathbb C: \Im z/\lambda\in (0, 1)\big\}$ with the properties
\begin{enumerate}
 \item $\underset{t\to +0}\Ltwolim\ \Psi(\cdot+ \mathrm it\lambda)= \psi(\cdot)$;
 \item there exists $\underset{t\to 1-0}{\Ltwolim}\ \Psi(\cdot+ \mathrm it\lambda)$;
 \item $\sup\limits_{t\in (0, 1)}\displaystyle\int_{\mathbb R}\big|\Psi(s+ \mathrm it\lambda)\big|^2\mathrm ds< \infty$.
\end{enumerate}
\end{definition}

For $\lambda\in \mathbb R$ let the operator of multiplication by $r^\lambda$ in $\mathsf L^2(\mathbb R_+, \textrm dr)$ be defined on its maximal domain $\mathsf L^2\big(\mathbb R_+, (1+ r^{2\lambda})\mathrm dr\big)$.
Applying the lemma of \cite{Titchmarsh} (Section 5.4, page 125) to justify the translations of the integration contour between different values of $t$ under Assumption 3 of Definition \ref{D lambda definition} we obtain
\begin{thm}\label{Mellin shift theorem}
Let $\lambda\in \mathbb R\setminus\{0\}$. Then the identity
\begin{equation*}
 \mathfrak D^\lambda =\mathcal M\mathsf L^2\big(\mathbb R_+, (1+ r^{2\lambda})\mathrm dr\big)
\end{equation*}
holds, and for any $\psi\in \mathfrak D^\lambda$ the function $\Psi$ from Definition \ref{D lambda definition} satisfies
\begin{equation*}
 \Psi(z) =(\mathcal Mr^{\Im z}\mathcal M^*\psi)(\Re z), \qquad\textrm{for all}\ z\in\mathfrak S^\lambda.
\end{equation*}
\end{thm}

We conclude that $r^\lambda$ acts as a complex shift in the Mellin space. Indeed, for $\lambda\in \mathbb R$ let $R^\lambda: \mathfrak D^\lambda\to \mathsf L^2(\mathbb R)$ be the linear operator defined by
\begin{equation*}
 R^\lambda\psi:= \begin{cases}\underset{t\to 1-0}{\Ltwolim}\ \Psi(\cdot+ \mathrm it\lambda), &\lambda \neq 0;\\ \psi, &\lambda =0,\end{cases}
\end{equation*}
with $\Psi$ as in Definition \ref{D lambda definition}.
It follows from Theorem \ref{Mellin shift theorem} that $R^\lambda$ is well-defined and that
\begin{equation}\label{complex shifted}
 \mathcal Mr^\lambda\mathcal M^*= R^\lambda
\end{equation}
holds (see \cite{YaouancOliverRaynal}, Section II).

The following lemma will be needed later:
\begin{lem}\label{Mellin-Bessel lemma}
 Let $J_m$ be the Bessel function with $m\in \mathbb Z$. The relation
\begin{equation*}
 \bigg(\mathcal M\Big((-\mathrm i)^m\int_0^\infty\sqrt{\cdot r}J_m(\cdot r)\psi(r)\mathrm dr\Big)\bigg)(s) = \Xi_m(s)(\mathcal M\psi)(-s)
\end{equation*}
holds for every $\psi\in C_0^\infty\big([0, \infty)\big)$ and $s\in\mathbb R$ with
\begin{equation}\label{Xi}
 \Xi_m(s):= (-\mathrm i)^{|m|}2^{-\mathrm is}\dfrac{\Gamma\Big(\big(|m|+ 1- \mathrm is\big)/2\Big)}{\Gamma\Big(\big(|m|+ 1+ \mathrm is\big)/2\Big)}.
\end{equation}
\end{lem}

\begin{proof}
It is enough to prove the statement for $m\in\mathbb N_0$, since $J_{-m}= (-1)^mJ_m$, see 10.4.1 in \cite{dlmf}.
According to 10.22.43 in \cite{dlmf},
\begin{equation}\label{Mellin of J_m}
 \lim_{R\to \infty}(-\mathrm i)^{m}\int_0^R t^{-\mathrm is}J_m(t)\,\mathrm dt= \Xi_m(s).
\end{equation}
It follows that
\begin{equation*}
 \sup_{L> 0}\Big|\int_0^Lt^{-\mathrm is}J_m(t)\,\mathrm dt\Big|< \infty.
\end{equation*}
The claim now follows from the representation
\begin{align*}
 &\bigg(\mathcal M\Big((-\mathrm i)^{m}\int_0^\infty\sqrt{\cdot r}J_m(\cdot r)\psi(r)\mathrm dr\Big)\bigg)(s)\\ &=\lim_{R\to\infty}\frac{(-\mathrm i)^m}{\sqrt{2\pi}}\int_0^Rp^{-\mathrm is}\int_{\supp\psi}\sqrt rJ_m(pr)\psi(r)\,\mathrm dr\,\mathrm dp
\end{align*}
by Fubini's theorem, dominated convergence and \eqref{Mellin of J_m}.
\end{proof}

\begin{rem}\label{Xi bar and inverse remark}
For any $m \in\mathbb Z$ the function $\Xi_m$ introduced in \eqref{Xi} allows an analytic continuation to $\mathbb C\setminus \Big(-\mathrm i\big(1 +|m| +2\mathbb N_0\big)\Big)$, whereas
\begin{align}\label{Xi bar and inverse}
 \Xi_m^{-1}(\cdot) =\overline{\Xi_m(\overline\cdot)}
\end{align}
allows an analytic continuation to $\mathbb C\setminus \Big(\mathrm i\big(1 +|m| +2\mathbb N_0\big)\Big)$.
\end{rem}

\begin{lem}\label{R and Xi commutation lemma}
For $(m, \lambda)\in \mathbb Z\times[0, 1]$ and any $\psi\in\mathfrak D^\lambda\supset \mathfrak D^1$ with
\begin{align}\label{Xi multiplication condition}
\Xi_m^{-1}\psi =\overline{\Xi_m(\overline\cdot)}\psi\in \mathfrak D^\lambda
\end{align}
the commutation rule
\begin{align}\label{R and Xi commutation}
 R^\lambda\Xi_m^{-1}\psi =\Xi_m^{-1}(\cdot +\mathrm i\lambda)R^\lambda\psi
\end{align}
applies. Except for $(m, \lambda) =(0, 1)$ condition \eqref{Xi multiplication condition} is automatically fulfilled for all $\psi\in\mathfrak D^\lambda$.
\end{lem}

\begin{proof}
 It follows from Remark \ref{Xi bar and inverse remark} that $\Xi_m^{-1}$ is analytic in $\mathfrak S^1$ and, for $(m, \lambda) \neq(0, 1)$, in a complex neighbourhood of $\mathfrak S^\lambda$. With the help of the Stirling asymptotic formula
\begin{equation}\label{Gamma asymptotics}
 \Gamma(z) =\sqrt{\frac{2\pi}z}\Big(\frac z{\mathrm e}\Big)^z\Big(1+ O\big(|z|^{-1}\big)\Big) \quad\textrm {for all }z\in\mathbb C\text{ with }|\arg z|<\pi -\delta,\quad \delta >0
\end{equation}
(see e.g. 5.11.3 in \cite{dlmf}) we conclude that the asymptotics
\begin{equation}\label{Xi asymptotics}
 \big|\Xi_m^{-1}(z)\big| =\big|\overline{\Xi_m(\overline{z})}\big| =|\Re z|^{-\Im z}\Big(1 +O\big(|z|^{-1}\big)\Big)\ \textrm{holds for }z\in\mathfrak S^1\ \textrm{as }|z|\to \infty.
\end{equation}
This implies that 
\begin{align}\label{Xi boundedness statement}
\Xi_m^{-1}\textrm{ is analytic and bounded in }\mathfrak S^\lambda\textrm{ for all }(m, \lambda)\in \big(\mathbb Z\times [0, 1]\big)\setminus \big\{(0, 1)\big\}, 
\end{align}
and the last statement of the lemma follows.

Since $\psi\in\mathfrak D^\lambda$, there exists $\Psi$ as in Definition \ref{D lambda definition}. Analogously, by \eqref{Xi multiplication condition} there exists $\Phi$ analytic in $\mathfrak S^\lambda$ corresponding to $\varphi:= \Xi_m^{-1}\psi$ as in Definition \ref{D lambda definition}. Then $\varphi$, $\psi\in \mathfrak D^{\lambda/2}$ and by \eqref{Xi boundedness statement}
\begin{align*}
 \Phi(\cdot +\mathrm i\lambda/2) =R^{\lambda/2}\varphi =R^{\lambda/2}\Xi_m^{-1}\psi =\Xi_m^{-1}(\cdot +\mathrm i\lambda/2)\Psi(\cdot +\mathrm i\lambda/2)
\end{align*}
holds on $\mathbb R$.
Thus $\Phi$ and $\Xi_m^{-1}\Psi$ must coincide on their joint domain of analyticity $\mathfrak S^\lambda$. Since $R^\lambda\Xi_m^{-1}\psi =\underset{t\to 1-0}\Ltwolim\ \Phi(\cdot+ \mathrm it\lambda)$ exists, it must coincide as a function on $\mathbb R$ with
\begin{align}\label{L2 limits}
\underset{t\to 1-0}\Ltwolim\ \Xi_m^{-1}(\cdot +\mathrm it\lambda)\Psi(\cdot+ \mathrm it\lambda) =\Xi_m^{-1}(\cdot +\mathrm i\lambda)\ \underset{t\to 1-0}\Ltwolim\ \Psi(\cdot+ \mathrm it\lambda) =\Xi_m^{-1}(\cdot +\mathrm i\lambda)R^\lambda\psi,
\end{align}
where the first equality in \eqref{L2 limits} can be justified by passing to an almost everywhere convergent subsequence.
\end{proof}

For $\lambda =1$, multiplying \eqref{R and Xi commutation} by $\Xi_m$ we conclude
\begin{cor}\label{V corollary}
For $m\in \mathbb Z$ and $\psi\in\mathfrak D^1$ (satisfying $\Xi_0^{-1}\psi\in \mathfrak D^1$ if $m =0$) the identity
\begin{align*}
 \Xi_m R^1\Xi_m^{-1}\psi =V_{|m|- 1/2}(\cdot +\mathrm i/2)R^1\psi
\end{align*}
holds with
\begin{equation}\label{V}
 V_j(z):=\frac{\Gamma\big((j+1+ \mathrm iz)/2\big)\Gamma\big((j+1- \mathrm iz)/2\big)}{2\Gamma\big((j+2+ \mathrm iz)/2\big)\Gamma\big((j+2- \mathrm iz)/2\big)},
\end{equation}
for $j\in\mathbb N_0- 1/2$ and $z\in \mathbb C\setminus \mathrm i(\mathbb Z+ 1/2)$.
\end{cor}

We will need the following properties of $V_j$:

\begin{lem}\label{V monotonicity lemma}
For every $j\in \mathbb N_0- 1/2$ the function \eqref{V} is analytic in $\mathbb C\setminus \mathrm i(\mathbb Z+ 1/2)$ and has the following properties:
\begin{enumerate}
 \item $V_j(z) =V_j(-z)$, for all $z\in \mathbb C\setminus \mathrm i(\mathbb Z+ 1/2)$;
 \item $V_j(s)$ is positive and strictly monotonously decreasing for $s\in \mathbb R_+$;
 \item $V_j(\mathrm i\zeta)$ is positive and strictly monotonously increasing for $\zeta\in [0, 1/2)$;
 \item The relation
\begin{equation}\label{V via 1/V}
 \big(z^2 +(j +1)^2\big)V_{j}(z) = \big(V_{j +1}(z)\big)^{-1}
\end{equation}
holds for all $z\in \mathbb C\setminus \mathrm i(\mathbb Z+ 1/2)$.
\end{enumerate}
\end{lem}

\begin{proof}
\emph{2.} For $z\in\mathbb C\setminus(-\mathbb N_0)$ let \textpsi$(z):= \Gamma'(z)/\Gamma(z)$ be the digamma function. Differentiating \eqref{V} and using Formula 5.7.7 in \cite{dlmf} we obtain
\begin{equation*}\begin{split}
 V'_j(s)&= V_j(s)\Im\Big(\text\textpsi\big((j+2+ \mathrm is)/2\big) -\text\textpsi\big((j+1+ \mathrm is)/2\big)\Big)\\ &= 2sV_j(s) \sum_{k =0}^\infty\frac{(-1)^{k +1}}{s^2 +(k +j +1)^2} < 0, \quad\text{for all }s >0.
\end{split}
\end{equation*}
\emph{3.} Analogously to \emph{2}, we compute
\begin{equation*}
 \mathrm iV'_j(\mathrm i\zeta)= 2\zeta V_j(\mathrm i\zeta)\sum_{k =0}^\infty\frac{(-1)^{k}}{(k +j +1)^2 -\zeta^2} > 0, \quad\text{for all }\zeta\in [0, 1/2).
\end{equation*}
\emph{4.} Follows directly from \eqref{V} and the recurrence relation $\Gamma(z+1)=z\Gamma(z)$ (valid for all $z\in\mathbb C\setminus (-\mathbb N_0)$).
\end{proof}

\paragraph{Angular decomposition.}

We can represent arbitrary $u\in \mathsf L^2(\mathbb R^2)$ in the polar coordinates as
\begin{equation*}
 u(r, \theta)= \frac1{\sqrt{2\pi}}\sum_{m\in \mathbb Z}r^{-1/2}u_m(r)\mathrm e^{\mathrm im\theta}
\end{equation*}
with
\begin{equation*}
 u_m(r):= \sqrt{\frac r{2\pi}}\int_0^{2\pi}u(r,\theta)\mathrm e^{-\mathrm im\theta}\mathrm d\theta.
\end{equation*}
The map
\begin{equation}\label{W}
 \mathcal W:\mathsf L^2(\mathbb R^2)\to \underset{m\in \mathbb Z}\bigoplus\mathsf L^2(\mathbb R_+), \quad u\mapsto \underset{m\in \mathbb Z}\bigoplus u_m
\end{equation}
is unitary.

For the proof of the following lemma (based on Lemmata 2.1, 2.2 of \cite{Bouzouina}) see the proof of Theorem 2.2.5 in \cite{BalinskyEvans}.
\begin{lem}\label{WF Coulomb lemma}
For $m\in \mathbb Z$ and $z\in (1, \infty)$ let
\begin{equation}\label{Q}
 Q_{|m|- 1/2}(z):= 2^{-|m|-1/2}\int_{-1}^1(1- t^2)^{|m|- 1/2}(z -t)^{-|m|-1/2}\,\mathrm dt
\end{equation}
be the Legendre function of second kind, see \cite{WhittakerWatson}, Section 15.3.
Let the quadratic form $\mathfrak q_m$ be defined on $\mathsf L^2\big(\mathbb R_+, (1+ p^2)^{1/2}\mathrm dp\big)$ by
\begin{equation*}
 \mathfrak q_m[g] := \pi^{-1}\iint_{\mathbb R_+^2}\overline{g(p)}Q_{|m|- 1/2}\bigg(\frac12\Big(\frac {q}p +\frac p{q}\Big)\bigg)g(q)\,\mathrm dq\,\mathrm dp.
\end{equation*}
Then for every $f$ in the Sobolev space $\mathsf H^{1/2}(\mathbb R^2)$ the relation
\begin{equation*}
 \int_{\mathbb R^2}|\mathbf x|^{-1}\big|f(\mathbf x)\big|^2\,\mathrm d\mathbf x = \sum_{m\in \mathbb Z}\mathfrak q_m\big[(\mathcal F f)_m\big]
\end{equation*}
holds. 
\end{lem}

The natural Hilbert space for spin-$1/2$ particles is $\mathsf L^2(\mathbb R^2, \mathbb C^2)$. Moreover, the natural angular momentum decomposition associated to \eqref{Dirac symbol} is not given by \eqref{W}, but by
\begin{equation*}
 \mathcal A:= \mathcal{SW}\begin{pmatrix}
                  1& 0\\ 0&\mathrm i
                 \end{pmatrix},
\end{equation*}
 where the unitary operator $\mathcal S$ is defined as
\begin{equation}\label{S}
 \mathcal S: \underset{m\in \mathbb Z}\bigoplus\mathsf L^2(\mathbb R_+, \mathbb C^2)\to \underset{\varkappa\in \mathbb Z+1/2}\bigoplus\mathsf L^2(\mathbb R_+, \mathbb C^2), \quad \underset{m\in \mathbb Z}\bigoplus\binom{\varphi_m}{\psi_m} \mapsto \underset{\varkappa\in \mathbb Z+1/2}\bigoplus\binom{\varphi_{\varkappa-1/2}}{\psi_{\varkappa+1/2}}.
\end{equation}

For $\nu\in [0, 1/2]$ and $\varkappa\in\mathbb Z+1/2$ we define the operators $D_{\varkappa,\, \max}^\nu$ in $\mathsf L^2(\mathbb R_+, \mathbb C^2)$ by the differential expressions
\begin{equation}\label{d_kappa}
 d^\nu_\varkappa= \begin{pmatrix}
                   -\dfrac\nu r& -\dfrac{\mathrm d}{\mathrm dr}- \dfrac\varkappa r\\ \dfrac{\mathrm d}{\mathrm dr}- \dfrac\varkappa r & -\dfrac\nu r
                  \end{pmatrix}
\end{equation}
on their maximal domains 
\begin{equation*}
\mathfrak D(D^\nu_{\varkappa,\, \max}):= \big\{u\in \mathsf L^2(\mathbb R_+, \mathbb C^2)\cap\mathsf{AC}_{\loc}(\mathbb R_+, \mathbb C^2): d^\nu_\varkappa u\in \mathsf L^2(\mathbb R_+, \mathbb C^2)\big\}.
\end{equation*}

Let $D^\nu_{\max}$ be the maximal operator in $\mathsf L^2(\mathbb R^2, \mathbb C^2)$ corresponding to \eqref{Dirac symbol} on the domain
\begin{align*}
 \mathfrak D(D^\nu_{\max}):= \Big\{&u\in \mathsf L^2(\mathbb R^2, \mathbb C^2): \text{ there exists }w\in \mathsf L^2(\mathbb R^2, \mathbb C^2)\text{ such that }\\ &\langle u,d^\nu v\rangle =\langle w, v\rangle\text{ holds for all }v\in \mathsf C_0^\infty\big(\mathbb R^2\setminus\{0\}, \mathbb C^2\big)\Big\}.
\end{align*}

The following Lemma follows from Section 7.3.3 in \cite{Thaller}.

\begin{lem}\label{angular Dirac lemma}
The operator $D^\nu_{\max}$ preserves the fibres of the half-integer angular momentum decomposition and satisfies
\begin{equation*}
 \mathcal A\,D^\nu_{\max}\,\mathcal A^*= \underset{\varkappa\in \mathbb Z+1/2}\bigoplus D^\nu_{\varkappa,\, \max}.
\end{equation*}
\end{lem}

\noindent In the following lemma we construct particular self-adjoint restrictions of $D^\nu_{\varkappa,\,\max}$.

\begin{lem}\label{D^nu_kappa lemma}
For $\nu \in[0, 1/2]$ and $\varkappa\in (\mathbb Z+ 1/2)$ let
\begin{equation}\label{operator core}
 \mathfrak C^\nu_\varkappa:= \mathsf C_0^\infty(\mathbb R_+, \mathbb C^2) \dot + \begin{cases}
                                                                        \Span\{\psi^\nu_\varkappa\}, & \textrm{for } \varkappa =\pm1/2 \textrm{ and }\nu\in(0, 1/2];\\ \{0\}, & \textrm{otherwise},
                                                                       \end{cases}
\end{equation}
with
\begin{equation}\label{extra functions}
 \psi^\nu_{\varkappa}(r):= \sqrt{2\pi}\binom{\nu}{\sqrt{\varkappa^2 -\nu^2} -\varkappa}r^{\sqrt{\varkappa^2 -\nu^2}}\mathrm e^{-r}, \qquad r\in\mathbb R_+.
\end{equation}
Then the restriction of $D^\nu_{\varkappa,\,\max}$ to $\mathfrak C^\nu_\varkappa$ is essentially self-adjoint in $\mathsf L^2(\mathbb R_+, \mathbb C^2)$. We define $D^\nu_{\varkappa}$ to be the self-adjoint operator in $\mathsf L^2(\mathbb R_+, \mathbb C^2)$ obtained as the closure of this restriction.
\end{lem}

\begin{proof}
For $\nu\in [0, 1/2]$, $\varkappa\in \mathbb Z+1/2$ let $D^\nu_{\varkappa,\,\min}$ be the closure of the restriction of $D^\nu_{\varkappa,\,\max}$ to $\mathsf C_0^\infty(\mathbb R_+, \mathbb C^2)$.
To determine the defect indices of $D^\nu_{\varkappa,\,\min}$ we observe that the fundamental solution of the equation $d^\nu_\varkappa\varphi =0$ in $\mathbb R_+$ is a linear combination of two functions:
\begin{equation*}
 \varphi^\nu_{\varkappa,\,\pm}(r):= \begin{cases}\displaystyle\binom{1/2 \pm1/2}{1/2 \mp1/2}r^{\pm\varkappa},& \textrm{ for }\nu =0;\\ \displaystyle\binom{\nu}{\pm\sqrt{\varkappa^2 -\nu^2} -\varkappa}r^{\pm\sqrt{\varkappa^2 -\nu^2}},& \textrm{ for }0 <\nu^2 <\varkappa^2;\end{cases}
\end{equation*}
\begin{equation*}
 \varphi^\nu_{\varkappa,\,+} := \binom{\nu}{-\varkappa}\quad \textrm{and}\quad \varphi^\nu_{\varkappa,\,0}(r):= \binom{\nu\ln r}{1 -\varkappa\ln r}, \quad\textrm{for }\nu^2= \varkappa^2 =1/4.
\end{equation*}
Now we apply Theorems 1.4 and 1.5 of \cite{Weidmann1971}.
Since $\varphi^\nu_{\varkappa,\,+}\not\in \mathsf L^2\big((1, \infty)\big)$ for any $\varkappa$ and $\nu$, the differential expression \eqref{d_kappa} is in the limit point case at infinity. 
For $\varkappa^2 -\nu^2\geqslant 1/4$ we have $\varphi^\nu_{\varkappa,\,-}\not\in \mathsf L^2\big((0, 1)\big)$ and hence \eqref{d_kappa} is in the limit point case at zero. In this case the defect indices of $D^\nu_{\varkappa,\,\min}$ are zero and thus $D^\nu_{\varkappa,\,\min}$ is self-adjoint.

For $\varkappa^2 -\nu^2< 1/4$, i.e. $\varkappa =\pm1/2$ and $\nu \in(0, 1/2]$,  any solution of $d^\nu_\varkappa\varphi =0$ belongs to $\mathsf L^2\big((0, 1)\big)$ and hence \eqref{d_kappa} is in the limit circle case at zero with the deficiency indices of $D^\nu_{\varkappa,\,\min}$ being $(1, 1)$. In this case every one-dimensional extension of $D^\nu_{\varkappa,\,\min}$ which is a restriction of $D^\nu_{\varkappa,\,\max}$ is self-adjoint (see e.g. \cite{BirmanSolomjak}, Section 4.4.1).
Theorem 1.5(2) in \cite{Weidmann1971} implies
\begin{equation}\label{general bc}
 \lim_{\varepsilon\to +0}\langle \varphi(\varepsilon), \mathrm i\sigma_2\psi(\varepsilon)\rangle_{\mathbb C^2} =0 \quad\textrm{for all }\varphi\in \mathfrak D(D^\nu_{\varkappa,\,\min}),\ \psi\in \mathfrak D(D^\nu_{\varkappa,\,\max}).
\end{equation}
Choosing $\psi(r) := \mathrm e^{-r}\varphi_{\varkappa, -}^\nu(r)$ for $\nu^2< \varkappa^2$ and $\psi(r) := \mathrm e^{-r}\varphi_{\varkappa, 0}^\nu(r)$ for $\nu^2= \varkappa^2$ in \eqref{general bc}, we conclude that $\psi^\nu_{\varkappa}\not\in\mathfrak D(D^\nu_{\varkappa,\,\min})$. Thus the closure of the restriction of $D^\nu_{\varkappa,\,\max}$ to $\mathfrak C^\nu_\varkappa$ is a one-dimensional extension of $D^\nu_{\varkappa,\,\min}$, hence a self-adjoint operator.
\end{proof}
For $\nu\in (0, 1/2]$ we now define
\begin{equation}\label{D^nu}
 D^\nu := \mathcal A^*\underset{\varkappa\in \mathbb Z+1/2}\bigoplus D^\nu_{\varkappa}\mathcal A.
\end{equation}
By Lemma \ref{angular Dirac lemma}, $D^\nu$ is a self-adjoint operator in $\mathsf L^2(\mathbb R^2, \mathbb C^2)$ corresponding to \eqref{Dirac symbol}.
Lemma \ref{D^nu_kappa lemma} implies
\begin{cor}\label{c: operator core}
Let $\delta_{\cdot, \cdot}$ be the Kronecker symbol. The set
\begin{equation*}
 \mathfrak C^\nu := \mathsf C_0^\infty\big(\mathbb R^2\setminus\{0\}, \mathbb C^2\big) \dot+ \Span\{\Psi^\nu_+, \Psi^\nu_-\}
\end{equation*}
with
\begin{equation*}
\begin{split}
 \Psi^\nu_\pm(r, \theta)&:= \Big(\mathcal A^*\bigoplus_{\varkappa\in \mathbb Z +1/2}\delta_{\varkappa, \pm1/2}\psi^\nu_{\pm1/2}\Big)(r, \theta)\\ &= \binom{\nu\mathrm e^{\mathrm i(\pm1/2 -1/2)\theta}}{-\mathrm i\big(\sqrt{1/4- \nu^2}\mp 1/2\big)\mathrm e^{\mathrm i(\pm1/2 +1/2)\theta}}r^{\sqrt{1/4- \nu^2}- 1/2}\mathrm e^{-r}
\end{split}\end{equation*}
is an operator core for $D^\nu$.
\end{cor}

\begin{rem}\label{canonicity remark}
For a particular class of non-semibounded operators a distinguished self-adjoint realisation can be selected by requiring the positivity of the Schur complement (see \cite{EstebanLoss}). In this sense $D^\nu +\diag(1, -1)$ is a distinguished self-adjoint realisation of the massive Coulomb-Dirac operator as proven in \cite{Mueller}.
\end{rem}

\paragraph{MWF-transform.}
We now introduce the unitary transform 
\begin{equation}\label{T}
\mathcal T:\mathsf L^2(\mathbb R^2)\to \underset{m\in \mathbb Z}\bigoplus\mathsf L^2(\mathbb R), \quad \mathcal T:= \mathcal M\mathcal W\mathcal F,
\end{equation}
where $\mathcal M$ acts fibre-wise.
A direct calculation using Lemmata \ref{Fourier channel lemma} and \ref{Mellin-Bessel lemma} gives
\begin{equation}\label{T in action}
\mathcal T\varphi= \underset{m\in \mathbb Z}\bigoplus\mathcal T_m\,\varphi_m,
\end{equation}
where for $m\in\mathbb Z$ the operators $\mathcal T_m: \mathsf L^2(\mathbb R_+)\to \mathsf L^2(\mathbb R)$ are given by
\begin{equation}\label{T components}
 (\mathcal T_m\phi)(s):= \Xi_m(s)(\mathcal M\phi)(-s) \quad \textrm{for any }\phi\in \mathsf L^2(\mathbb R_+).
\end{equation}

In the following two lemmata we study the actions of several operators in the MWF-representation.
\begin{lem}\label{transformed operators lemma}
 The relations
\begin{equation}\label{sigma*p with T}
 (\mathcal S\mathcal T)(-\mathrm i\boldsymbol\sigma\cdot\nabla)(\mathcal S\mathcal T)^*= \underset{\varkappa\in \mathbb Z+ 1/2}\bigoplus (R^1\otimes\sigma_1)
\end{equation}
and for any $\lambda\in \mathbb R$
\begin{equation}\label{|p| with T}
 \mathcal T(-\Delta)^{\lambda/2}\mathcal T^*= \underset{m\in \mathbb Z}\bigoplus R^\lambda
\end{equation}
hold.
\end{lem}
\begin{proof}
For any $(\varphi, \psi)\in \mathsf H^1(\mathbb R^2, \mathbb C^2)$, applying \eqref{T}, \eqref{W} and \eqref{S} we obtain
\begin{align*}
 &\mathcal{ST}(-\mathrm i\boldsymbol\sigma\cdot\nabla)\binom{\varphi}{\psi}= \mathcal{SM}\mathcal W\begin{pmatrix}
                               0& p\mathrm e^{-\mathrm i\omega}\\ p\mathrm e^{\mathrm i\omega}& 0
                              \end{pmatrix}\binom{\mathcal F\varphi}{\mathcal F\psi}\\ &= \underset{\varkappa\in \mathbb Z +1/2}\bigoplus \mathcal Mp\binom{(\mathcal F\psi)_{\varkappa +1/2}}{(\mathcal F\varphi)_{\varkappa -1/2}} =\bigg(\underset{\varkappa\in \mathbb Z+ 1/2}\bigoplus(\mathcal Mp\mathcal M^*)\otimes \sigma_1\bigg)\mathcal{SMWF}\binom\varphi\psi,
\end{align*}
which according to \eqref{complex shifted} and \eqref{T} leads to \eqref{sigma*p with T}. To get \eqref{|p| with T}, the same argument applies with $p^\lambda$ instead of $\begin{pmatrix}
                               0& p\mathrm e^{-\mathrm i\omega}\\ p\mathrm e^{\mathrm i\omega}& 0
                              \end{pmatrix}$ and $\mathcal S$ removed.
\end{proof}

\begin{lem}\label{1/r with T lemma}
The relation
\begin{equation}\label{1/r with T}
\mathcal T|\cdot|^{-1}\mathcal T^*= \underset{m\in \mathbb Z}\bigoplus \Xi_m R^1\Xi_m^{-1}
\end{equation}
holds.
\end{lem}

\begin{proof}
For any $\varphi\in \mathsf L^2\Big(\mathbb R^2, \big(1 +|\mathbf x|^{-2}\big)\,\mathrm d\mathbf x\Big)$ applying \eqref{T in action}, \eqref{T components}, \eqref{complex shifted} and \eqref{Xi bar and inverse} we obtain for almost every $s\in \mathbb R$
\begin{equation*}\begin{split}
 &\Big(\mathcal T\big(|\cdot|^{-1}\varphi\big)\!\Big)(s)= \underset{m\in \mathbb Z}\bigoplus\Xi_{m}(s)\Big(\mathcal M\big((\cdot)^{-1}\varphi_m\big)\!\Big)(-s)= \underset{m\in \mathbb Z}\bigoplus\Xi_{m}(s)(R^{-1}\mathcal M\varphi_{m})(-s)\\
 &= \underset{m\in \mathbb Z}\bigoplus\Xi_{m}(s)\Big(R^1\big((\mathcal M\varphi_{m})(-\cdot)\big)\Big)(s)= \underset{m\in \mathbb Z}\bigoplus\Xi_{m}(s)\big(R^1\Xi_m^{-1}\mathcal T_{m}\,\varphi_{m}\big)(s).
 \end{split}
\end{equation*}
This together with \eqref{T in action} gives \eqref{1/r with T}.
\end{proof}

\paragraph{U-transform.}

For $\varkappa\in \mathbb Z +1/2$ let the unitary operators $\mathcal U_\varkappa:\mathsf L^2(\mathbb R_+, \mathbb C^2)\to \mathsf L^2(\mathbb R, \mathbb C^2)$ be defined by
\begin{align}\label{Upsi}
 \mathcal U_\varkappa\binom{\psi_1}{\psi_2} :=\Bigg(\mathcal{STA}^*\bigoplus_{\widetilde\varkappa\in \mathbb Z +1/2}\delta_{\widetilde\varkappa, \varkappa}\binom{\psi_1}{\psi_2}\Bigg)_{\varkappa} =\binom{\mathcal T_{\varkappa -1/2}\psi_1}{-\mathrm i\mathcal T_{\varkappa +1/2}\psi_2}.
\end{align}

A straightforward calculation involving \eqref{T components}, \eqref{Mellin transform}, \eqref{V}, \eqref{V via 1/V} and the elementary properties of the gamma function delivers
\begin{lem}\label{T transformed extra functions lemma}
For $\nu\in (0, 1/2]$ let
\begin{equation*}
 \beta:= \sqrt{1/4 -\nu^2}.
\end{equation*}
The functions \eqref{extra functions} from the operator core $\mathfrak C^\nu_{\pm1/2}$ of $D^\nu_{\pm1/2}$ satisfy the relation
\begin{equation*}
 \mathcal U_{\pm1/2}\psi^\nu_{\pm1/2} =\chi_\pm^\nu\binom{1}{\nu V_{\mp1/2}(\mathrm i\beta)} + \binom{\xi_\pm^\nu}{\eta_\pm^\nu}
\end{equation*}
with
\begin{align}
 \chi_\pm^\nu&:= \nu\Xi_{\pm1/2- 1/2}\big(\mathrm i(\beta +1/2)\big)\frac{(\cdot -\mathrm i)\Gamma(\mathrm i\cdot+ \beta +1/2)}{\mathrm i(\beta -1/2)},\label{chi}\\
 \xi_\pm^\nu&:= \nu\Gamma(\mathrm i\cdot+ \beta+ 1/2)\Big(\Xi_{\pm1/2 -1/2}(\cdot) -\frac{(\cdot -\mathrm i)\Xi_{\pm1/2- 1/2}\big(\mathrm i(\beta +1/2)\big)}{\mathrm i(\beta -1/2)}\Big),\label{xi}\\
 \eta_\pm^\nu&:= \frac{\mathrm i\nu^2\Gamma(\mathrm i\cdot+ \beta+ 1/2)}{\beta \pm1/2}\Big(\Xi_{\pm1/2 +1/2}(\cdot) -\frac{(\cdot -\mathrm i)\Xi_{\pm1/2 +1/2}\big(\mathrm i(\beta +1/2)\big)}{\mathrm i(\beta -1/2)}\Big).\label{eta}
\end{align}
\end{lem}

\section{Fourier-Mellin theory of the relativistic massless Coulomb operator in two dimensions}\label{scalar section}

For $\alpha\in \mathbb R$ consider the symmetric operator
\begin{equation*}
 \widetilde H^\alpha :=(-\Delta)^{1/2} - \alpha|\cdot|^{-1}
\end{equation*}
in $\mathsf L^2(\mathbb R^2)$ on the domain 
\begin{equation*}
\mathfrak D(\widetilde H^\alpha):= \mathsf H^1(\mathbb R^2)\cap \mathsf L^2\big(\mathbb R^2, |\mathbf x|^{-2}\,\mathrm d\mathbf x\big). 
\end{equation*}
According to Lemmata \ref{transformed operators lemma} and \ref{1/r with T lemma} and Corollary \ref{V corollary} we have
\begin{equation}\label{MWF Herbst}
 \mathcal{T}\widetilde H^\alpha\mathcal{T}^*= \underset{m\in \mathbb Z}\bigoplus \big( 1 -\alpha V_{|m|- 1/2}(\cdot + \mathrm i/2)\big)R^1 =:\underset{m\in \mathbb Z}\bigoplus\widetilde H^\alpha_m,
\end{equation}
where the right hand side is an orthogonal sum of operators in $\mathsf L^2(\mathbb R)$ densely defined on
\begin{equation}\label{H m alpha tilde domain}
\mathfrak D(\widetilde H^\alpha_m):= \big\{\varphi\in \mathfrak D^1: \Xi_m^{-1}\varphi\in \mathfrak D^1\big\}.
\end{equation}
According to Lemma \ref{R and Xi commutation lemma}, $\mathfrak D(\widetilde H^\alpha_m)= \mathfrak D^1$ for $m\in \mathbb Z\setminus\{0\}$.
On the other hand, $\mathfrak D(\widetilde H^\alpha_0) \neq\mathfrak D^1$ for any $\alpha\neq 0$, which corresponds to the absence of Hardy inequality in two dimensions.

\begin{lem}\label{critical constants lemma}
 For $m\in \mathbb Z$ and $\alpha \in\mathbb R$ the operator $\widetilde H^\alpha_m$ is symmetric. It is bounded below (and non-negative) in $\mathsf L^2(\mathbb R)$ if and only if
\begin{equation}\label{alpha_c}
 \alpha\leqslant\alpha_m:= \frac1{V_{|m| -1/2}(0)}= \frac{2\Gamma^2\Big(\big(2|m| +3\big)/4\Big)}{\Gamma^2\Big(\big(2|m|+ 1\big)/4\Big)}.
\end{equation}
\end{lem}
\begin{proof}
For $m\in \mathbb Z$ and $\varphi\in \mathfrak D(\widetilde H^\alpha_m)\subset \mathfrak D^1$ using Corollary \ref{V corollary} and Lemma \ref{R and Xi commutation lemma} we compute
\begin{equation}\begin{split}\label{1}
\langle \varphi,\widetilde H_m^\alpha\varphi\rangle &=\langle\varphi, R^1\varphi\rangle -\alpha\big\langle\varphi, \Xi_m R^1\Xi_m^{-1}\varphi\big\rangle\\ &=\int_{-\infty}^{+\infty}\big( 1 -\alpha V_{|m|- 1/2}(s)\big)\big|(R^{1/2}\varphi)(s)\big|^2\mathrm ds.
\end{split}\end{equation}
Since the right hand side is real-valued, $\widetilde H^\alpha_m$ is symmetric.
By Lemma \ref{V monotonicity lemma}(1, 2) the condition \eqref{alpha_c} is equivalent to the non-negativity of $1 -\alpha V_{|m|- 1/2}(s)$ for all $s\in\mathbb R$. For $\alpha> \alpha_m$, $1 -\alpha V_{|m|- 1/2}$ is negative on an open interval $(-\rho, \rho)$ ($\rho$ depends on $m$ and $\alpha$). For $n\in\mathbb N$ the functions
\[
\varphi_n:= \widetilde\varphi_n/\|\widetilde\varphi_n\|_{\mathsf L^2(\mathbb R)}, \quad\textrm{with } \widetilde\varphi_n(s):= \mathrm e^{-n(s- \mathrm i/2)^2}\big(1 -\mathrm e^{-n^2(s- \mathrm i)^2}\big)
\]
are normalised in $\mathsf L^2(\mathbb R)$ and belong to $\mathfrak D(\widetilde H^\alpha_m)$. 
Using Lemma \ref{V monotonicity lemma} we estimate
\begin{equation}\label{V bound}
 1 -\alpha V_{|m|- 1/2} \leqslant \big(1 -\alpha V_{|m|- 1/2}(\rho/2)\big)\mathbbm1_{[-\rho/2, \rho/2]} + \mathbbm1_{\mathbb R\setminus (-\rho/2, \rho/2)}
\end{equation}
as a function on $\mathbb R$. It follows from \eqref{V bound} that for $n$ big enough \eqref{1} becomes negative with $\varphi:= \varphi_n$. But then replacing $\varphi_n(s)$ by $\lambda^{\mathrm is}\varphi_n(s)$ (still normalised and belonging to $\mathfrak D(\widetilde H^\alpha_m)$ for all $\lambda\in \mathbb R_+$) we can make the quadratic form \eqref{1} arbitrarily negative.
\end{proof}

Given $m\in \mathbb Z$, Lemma \ref{critical constants lemma} allows us for $\alpha\leqslant\alpha_m$ to pass from the symmetric operator $\widetilde H^\alpha_m$ to the self-adjoint operator $H^\alpha_m$ by Friedrichs extension \cite{Friedrichs}. 
The following description of the domains of $H^\alpha_m$ with $m\neq 0$, $\alpha \in (0, \alpha_m]$ follows analogously to Corollary 2 in \cite{YaouancOliverRaynal} (see also Section 2.2.3 of \cite{BalinskyEvans}) with the help of Lemma \ref{V monotonicity lemma}:
\begin{lem}\label{H alpha m domains lemma}
 Let $m\in\mathbb Z\setminus \{0\}$.
\begin{enumerate}
 \item For $\alpha < V^{-1}_{|m|-1/2}(\mathrm{i}/2)$ the operator $\widetilde H^\alpha_m$ is self-adjoint.
 \item For $\alpha = V^{-1}_{|m|-1/2}(\mathrm{i}/2)$ the operator $\widetilde H^\alpha_m$ is essentially self-adjoint.
 \item For $\alpha \in \big(V^{-1}_{|m|-1/2}(\mathrm{i}/2), \alpha_m\big]$ the Friedrichs extension $H^\alpha_m$ of $\widetilde H^\alpha_m$ is the restriction of
\begin{equation}\label{H_m^alpha^*}
 (\widetilde H_m^\alpha)^*= R^1\big(1 -\alpha V_{|m| -1/2}(\cdot\, -\mathrm i/2)\big)
\end{equation}
to
\begin{equation*}
 \mathfrak D(H^\alpha_m) =\mathfrak D^1\dot{+} \Span\big\{(\cdot\, -\mathrm i/2 +\mathrm i\zeta_{m, \alpha})^{-1}\big\},
\end{equation*}
where $\zeta_{m, \alpha}$ is the unique solution of
\begin{equation}\label{zeta}
 1- \alpha V_{|m| -1/2}(-\mathrm i\zeta_{m, \alpha}) =0
\end{equation}
in $(-1/2, 0]$.
\end{enumerate}
\end{lem}
In the case $m =0$ the functions $V_{-1/2}(\cdot\, \pm \mathrm i/2)$ are not bounded on $\mathbb R$, which makes the argument of \cite{YaouancOliverRaynal} not directly applicable (as both factors in \eqref{H_m^alpha^*} are unbounded).
Instead of providing an exact description of $\mathfrak D(H^\alpha_0)$ we prove a simpler result:
\begin{lem}\label{H alpha 0 domain lemma}
For $\alpha \in (0, \alpha_0]$ the domain of the Friedrichs extension $H^\alpha_0$ of $\widetilde H^\alpha_0$ satisfies
\begin{equation*}
 \mathfrak D(H^\alpha_0) \supseteq\mathfrak D(\widetilde H^\alpha_0)\dot{+} \Span\{\varphi_0^\alpha\}
\end{equation*}
with
\begin{equation*}
 \varphi_0^\alpha(s):= \frac{s -\mathrm i}{(s -2\mathrm i)(s -\mathrm i/2 +\mathrm i\zeta_{0, \alpha})}
\end{equation*}
and $\zeta_{0,\alpha}$ defined as in \eqref{zeta}.
Moreover,
\begin{equation}\label{extended action}
 (H^\alpha_0\varphi_0^\alpha)(s) =\frac{s\big(1 -\alpha V_{-1/2}(s +\mathrm i/2)\big)}{(s -\mathrm i)(s +\mathrm i/2 +\mathrm i\zeta_{0, \alpha})}
\end{equation}
holds for all $s\in \mathbb R\setminus \{0\}$.
\end{lem}

\begin{proof}
According to Theorem 5.38 in \cite{Weidmann1}, $H^\alpha_0$ is the restriction of $(\widetilde H^\alpha_0)^*$ to $\mathfrak D(H^\alpha_0):= \mathfrak Q_0^\alpha\cap \mathfrak D\big((\widetilde H^\alpha_0)^*\big)$, where $\mathfrak Q_0^\alpha$ is the closure of $\mathfrak D(\widetilde H^\alpha_0)$ in the norm of the quadratic form of $\widetilde H^\alpha_0 +1$. 

Since $\mathsf C_{0}^{\infty}\big(\mathbb{R}^2\setminus \{0\}\big)\subset \mathfrak D(\widetilde H^\alpha)$ is dense in $\mathsf H^{1/2}(\mathbb{R}^2)$, the representation \eqref{MWF Herbst} shows that $\mathfrak{D}(\widetilde H_{0}^{\alpha})$ is dense in $\mathfrak{D}^{1/2}$ with respect to the graph norm of $R^{1/2}$
for all $\alpha \in (0,\alpha_0]$.
Lemma \ref{critical constants lemma} implies the inequalities
\begin{equation*}
\begin{split}
\langle\varphi, R^1\varphi\rangle\geqslant \langle\varphi, \widetilde H^\alpha_0\varphi\rangle \geqslant (1 -\alpha/\alpha_0)\langle\varphi, R^1\varphi\rangle
\end{split}
\end{equation*}
for all $\alpha \in (0, \alpha_0]$ and $\varphi\in \mathfrak D(\widetilde H^\alpha_0)$.  Thus $\mathfrak Q_0^\alpha =\mathfrak{D}^{1/2} \subset \mathfrak Q_0^{\alpha_0}$ for $\alpha \in (0, \alpha_0)$
and the right hand side of \eqref{1} coincides with the closure of the quadratic form of $\widetilde H_{0}^{\alpha}$ on every $\varphi\in \mathfrak{D}^{1/2}$ for $\alpha \in (0, \alpha_0]$.

For $n\in \mathbb N$ let $\psi_n(s) :=(s -\mathrm i)(s -2\mathrm i)^{-1}(s -\mathrm i/2 -\mathrm i/n)^{-1}\in \mathfrak{D}^{1/2} \subset \mathfrak Q_0^{\alpha_0}$. Computing the right hand side of \eqref{1} on $\varphi:= \psi_n - \psi_m$ with $m\leqslant n$ we obtain
\begin{equation*}
\begin{split}
 &\int_{-\infty}^{+\infty}\big( 1 -\alpha_0 V_{- 1/2}(s)\big)\big|R^{1/2}(\psi_n- \psi_m)(s)\big|^2\mathrm ds\leqslant \int_{-\infty}^{+\infty}\frac{\big( 1 -\alpha_0 V_{- 1/2}(s)\big)}{s^2(m^2s^2+1)}\mathrm ds.
\end{split}
\end{equation*}
By Lemma \ref{V monotonicity lemma} and monotone convergence we conclude that $(\psi_n)_{n\in \mathbb N}$ is a Cauchy sequence in $\mathfrak Q_0^{\alpha_0}$ which converges to $\varphi_0^{\alpha_0}$ in $\mathsf L^2(\mathbb R)$. Thus $\varphi_0^{\alpha_0}$ belongs to $\mathfrak Q_0^{\alpha_0}$.

For every $\varphi \in\mathfrak D(\widetilde H^\alpha_0)$ with $\alpha \in (0, \alpha_0]$ taking into account the relations
\begin{align*}
 &\varphi_0^\alpha\in \mathfrak D^{1/4},\ \Xi_0^{-1}\varphi_0^\alpha\in \mathfrak D^{1/4}\text{ and }\big(1 -\alpha V_{- 1/2}(\cdot -\mathrm i/4)\big)\varphi_0^\alpha(\cdot +\mathrm i/4)\in \mathfrak D^{3/4}
\end{align*}
(recall \eqref{zeta} and \eqref{Xi asymptotics}) and using Corollary \ref{V corollary} and Lemma \ref{R and Xi commutation lemma} we obtain
\begin{align*}
 \langle \varphi_0^\alpha, \widetilde H^\alpha_0\varphi\rangle &= \Big\langle \varphi_0^\alpha, \big(R^1 -\alpha\Xi_0 R^1\Xi_0^{-1}\big)\varphi\Big\rangle\\& =\langle R^{1/4}\varphi_0^\alpha, R^{3/4}\varphi\rangle -\alpha\langle R^{1/4}\Xi_0^{-1}\varphi_0^\alpha, R^{3/4}\Xi_0^{-1}\varphi\rangle\\ & =\Big\langle\big(1 -\alpha V_{- 1/2}(\cdot -\mathrm i/4)\big)\varphi_0^\alpha(\cdot +\mathrm i/4), R^{3/4}\varphi\rangle\\ &=\Big\langle\big(1 -\alpha V_{- 1/2}(\cdot +\mathrm i/2)\big)\varphi_0^\alpha(\cdot +\mathrm i), \varphi\Big\rangle.
\end{align*}
It follows that $\varphi_0^{\alpha}\in \mathfrak D\big((\widetilde H^\alpha_0)^*\big)$ and \eqref{extended action} holds for all $\alpha \in (0, \alpha_0]$.
\end{proof}

We now make a crucial observation concerning the functions \eqref{extra functions} transformed in Lemma \ref{T transformed extra functions lemma}.

\begin{lem}\label{domain connection lemma}
 Let $\nu\in (0, 1/2]$. The functions \eqref{chi}, \eqref{xi} and \eqref{eta} satisfy:
\begin{enumerate}
  \item $\xi^\nu_\pm$ and $\eta^\nu_\pm$ belong to $\mathfrak D^1$;
  \item $\Xi_0^{-1}\xi^\nu_+$ and $\Xi_0^{-1}\eta^\nu_-$ belong to $\mathfrak D^1$;
  \item $\chi^\nu_\pm$ belong to $\mathfrak D\big(H_0^{(V_{-1/2}(\mathrm i\beta))^{-1}}\big)$ and
\begin{equation}\label{H_0 on chi}
  H_0^{(V_{-1/2}(\mathrm i\beta))^{-1}}\chi^\nu_\pm =\Big(1 -\big(V_{-1/2}(\mathrm i\beta)\big)^{-1}V_{-1/2}(\cdot\, +\mathrm i/2)\Big)\chi^\nu_\pm(\cdot\, +\mathrm i);
\end{equation}
  \item $\chi^\nu_\pm$ belong to $\mathfrak D\big(H_1^{(V_{1/2}(\mathrm i\beta))^{-1}}\big)$.
\end{enumerate}
\end{lem}

\begin{proof}
\emph{1.} By Remark \ref{Xi bar and inverse remark} and since the gamma function is analytic in $\mathbb C\setminus (-\mathbb N_0)$ with a simple pole at zero, $\xi^\nu_\pm$ and $\eta^\nu_\pm$ are analytic in a complex neighbourhood of the strip $\mathfrak S^1$. Thus, for every $\rho >0$, $\xi^\nu_\pm$ and $\eta^\nu_\pm$ are bounded on $\mathfrak A_\rho :=\big\{z\in\mathbb C: \Re z\in [-\rho, \rho], \Im z\in [0, 1]\big\}$. On $\mathfrak S^1\setminus \mathfrak A_\rho$ substituting the asymptotics \eqref{Gamma asymptotics} into \eqref{xi}, \eqref{eta} and \eqref{Xi} (or using \eqref{Xi bar and inverse} and \eqref{Xi asymptotics}) and choosing $\rho$ big enough we obtain the properties \emph{1.--3.} of Definition \ref{D lambda definition}.

\emph{2.} Both $\Xi_0^{-1}\xi^\nu_+$ and $\Xi_0^{-1}\eta^\nu_-$ are analytic in a complex neighbourhood of $\mathfrak S^1$. We can thus repeat the proof of \emph{1.} taking \eqref{Xi asymptotics} into account.

\emph{3.} By Lemma \ref{H alpha 0 domain lemma}, it suffices to show that
\begin{equation}\label{chi-phi}
 \chi^\nu_{\pm} +\mathrm i\nu\frac{2\beta -3}{2\beta -1}\Xi_{\pm1/2 -1/2}\big(\mathrm i(\beta+ 1/2)\big)\varphi_0^{(V_{-1/2}(\mathrm i\beta))^{-1}}\in\mathfrak D\Big(\widetilde H^{(V_{-1/2}(\mathrm i\beta))^{-1}}_0\Big),
\end{equation}
see \eqref{H m alpha tilde domain}. This follows analogously to \emph{1,} since
$\zeta_{0, (V_{-1/2}(\mathrm i\beta))^{-1}} :=-\beta$ is the solution of \eqref{zeta}. Formula \eqref{H_0 on chi} follows from \eqref{chi-phi}, \eqref{MWF Herbst} and \eqref{extended action}.

\emph{4.} The proof is analogous to \emph{3.} Since $\zeta_{1, (V_{1/2}(\mathrm i\beta))^{-1}} :=-\beta$ is the solution of \eqref{zeta} we conclude that
\begin{equation*}
 \chi^\nu_{\pm} +\mathrm i\nu\Xi_{\pm1/2 -1/2}\big(\mathrm i(\beta+ 1/2)\big)\big(\cdot -\mathrm i/2 +\mathrm i\zeta_{1, (V_{1/2}(\mathrm i\beta))^{-1}}\big)^{-1}
\end{equation*}
belongs to $\mathfrak D\Big(\widetilde H^{(V_{1/2}(\mathrm i\beta))^{-1}}_1\Big)$ characterised in Lemma \ref{H alpha m domains lemma}.
\end{proof}

\section{Critical channels estimate}\label{critical channels section}

For $\nu\in (0, 1/2]$ we introduce the $(2\times 2)$-matrix-valued function on $\mathbb R$:
\begin{equation*}
 M_{\pm}^\nu(s):= \begin{pmatrix}-\nu V_{\mp1/2}(s +\mathrm i/2) & 1\\ 1 & -\nu V_{\pm1/2}(s +\mathrm i/2)\end{pmatrix}.
\end{equation*}

\begin{lem}\label{maximal subdomain transformed lemma}
For any $\Psi\in\mathfrak C^\nu_{\pm1/2}$ there exists a decomposition
\begin{align}\label{U Psi representation}
 \mathcal U_{\pm1/2}\Psi =\binom{\zeta}\upsilon +a\chi_\pm^\nu\binom1{\nu V_{\mp1/2}(\mathrm i\beta)}
\end{align}
with $\zeta\in \mathfrak D\Big(\widetilde H_{1/2\mp1/2}^{(V_{\mp1/2}(\mathrm i\beta))^{-1}}\Big)$, $\upsilon\in \mathfrak D\Big(\widetilde H_{1/2\pm1/2}^{(V_{\pm1/2}(\mathrm i\beta))^{-1}}\Big)$ and $a\in \mathbb C$.
Moreover, the representation
\begin{equation*}
\begin{split}
 \mathcal U_{\pm1/2}D^\nu_{\pm1/2}\Psi = M^\nu_{\pm}\binom{R^1\zeta +a\chi_\pm^\nu(\cdot +\mathrm i)}{R^1\upsilon +a\nu V_{\mp1/2}(\mathrm i\beta)\chi_\pm^\nu(\cdot +\mathrm i)}
\end{split}\end{equation*}
holds.
\end{lem}

\begin{proof}
The decomposition \eqref{U Psi representation} follows from \eqref{operator core}, Lemma \ref{T transformed extra functions lemma}, \eqref{H m alpha tilde domain} and Lemma \ref{domain connection lemma}.
 For any $(\varpi, \varsigma)\in \mathsf C^\infty_0\big(\mathbb R_+, \mathbb C^2\big)$ using \eqref{D^nu}, \eqref{Upsi}, Lemmata \ref{transformed operators lemma} and \ref{1/r with T lemma}, Corollary \ref{V corollary}, \eqref{H_0 on chi} and \eqref{H_m^alpha^*} we obtain
\begin{equation*}
\begin{split}
 &\Big\langle D^\nu_{\pm1/2}\Psi, \binom{\varpi}{\varsigma}\Big\rangle =\Big\langle \Psi, D^\nu_{\pm1/2}\binom{\varpi}{\varsigma}\Big\rangle \\ & =\Big\langle \Psi, \bigg(\mathcal A(\mathcal{ST})^*\mathcal{ST}D^\nu(\mathcal{ST})^*\mathcal{ST}\mathcal A^*\bigoplus_{\varkappa\in \mathbb Z+ 1/2}\delta_{\varkappa, \pm1/2}\binom{\varpi}{\varsigma}\bigg)_{\varkappa =\pm1/2}\Big\rangle \\ & =\Big\langle \binom{\zeta}\upsilon +a\chi_\pm^\nu\binom1{\nu V_{\mp1/2}(\mathrm i\beta)},\\ &\hspace{1mm} \begin{pmatrix}-\nu\Xi_{1/2\mp1/2}R^1\Xi_{1/2\mp1/2}^{-1} & R^1\\ R^1 & -\nu\Xi_{1/2\pm1/2}R^1\Xi_{1/2\pm1/2}^{-1}\end{pmatrix}\mathcal U_{\pm1/2}\binom{\varpi}{\varsigma}\Big\rangle \\ & =\Big\langle M_\pm^\nu R^1\binom{\zeta}\upsilon, \mathcal U_{\pm1/2}\binom{\varpi}{\varsigma}\Big\rangle\\ & +a\Big\langle\chi_\pm^\nu\binom{\nu V_{\mp1/2}(\mathrm i\beta)}1, \begin{pmatrix}H_{1/2\mp1/2}^{(V_{\mp1/2}(\mathrm i\beta))^{-1}} & 0\\ 0 & H_{
1/2\pm1/2}^{(V_{\pm1/2}(\mathrm 
i\beta))^{-1}}\end{pmatrix}\mathcal U_{\pm1/2}\binom{\varpi}{\varsigma}\Big\rangle\\ & = \Big\langle \mathcal U_{\pm1/2}^*M^\nu_{\pm}\binom{R^1\zeta +a\chi_\pm^\nu(\cdot +\mathrm i)}{R^1\upsilon +a\nu V_{\mp1/2}(\mathrm i\beta)\chi_\pm^\nu(\cdot +\mathrm i)}, \binom{\varpi}{\varsigma}\Big\rangle.
\end{split}\end{equation*}
By density of $\mathsf C_0^\infty\big(\mathbb R_+, \mathbb C^2)$ the claim follows.
\end{proof}

\begin{lem}\label{critical Dirac via Kato bound lemma}
For $\nu\in (0, 1/2]$ define the functions
\begin{align*}
 K^\nu_\pm(s):= \Big| 1 - \big(V_{\pm1/2}(\mathrm i\beta)\big)^{-1}V_{\pm1/2}(s +\mathrm i/2)\Big|^2
\end{align*}
on $\mathbb R\setminus\{0\}$.
Then there exists a constant $\eta_\nu >0$ such that the lower bound
\begin{align}\label{critical Dirac via Kato bound}
 (M_{\pm}^\nu)^*M_{\pm}^\nu \geqslant\eta_\nu\diag(K^\nu_\mp, K^\nu_\pm)
\end{align}
holds point-wise on $\mathbb R\setminus\{0\}$.
\end{lem}

\begin{proof}
It is enough to establish \eqref{critical Dirac via Kato bound} for $M_{+}^\nu$ and then use the relation $M_{-}^\nu= \sigma_1M_{+}^\nu\sigma_1$. We introduce the shorthand $V:= V_{1/2}(\mathrm i\beta) =\nu^{-2}\big(V_{-1/2}(\mathrm i\beta)\big)^{-1}$ (see  \eqref{V via 1/V} for the second equality).
For any $s\in \mathbb R\setminus\{0\}$, estimating
\begin{align*}
 K^\nu_\pm(s)\leqslant 2\Big(1 + \big(V_{\pm1/2}(\mathrm i\beta)\big)^{-2}\big|V_{\pm1/2}(s +\mathrm i/2)\big|^2\Big)
\end{align*}
and using \eqref{V} we obtain
\begin{align*}
 K^\nu_+(s)\leqslant 2\big(1 + (1+ s^2)^{-1}V^{-2}\big)
\end{align*}
and
\begin{align*}
 K^\nu_-(s)\leqslant 2(1 + \nu^4V^{2}s^{-2}).
\end{align*}
Analogously we get
\begin{align*}
\big(M_{+}^\nu(s)\big)^*M_{+}^\nu(s)= \begin{pmatrix}
                                                 1 + \nu^2s^{-2}& -\dfrac{\nu(1 -2\mathrm is)}{s^2 +\mathrm is}P(s)\\ -\dfrac{\nu(1 +2\mathrm is)}{s^2 -\mathrm is}\overline P(s) & 1+\nu^2(1+ s^2)^{-1}
                                                \end{pmatrix}
\end{align*}
with
\begin{align*}
 P(s):= \frac{\Gamma\big((1 +\mathrm is)/2\big)\Gamma(-\mathrm is/2)}{\Gamma\big((1 -\mathrm is)/2\big)\Gamma(\mathrm is/2)}, \qquad \big|P(s)\big| =1.
\end{align*}
Thus for any $\eta >0$ the inequality
\begin{align}\label{determinant}
 \det\Big(\big(M_{+}^\nu(s)\big)^*M_{+}^\nu(s) -\frac\eta2\diag\big(K^\nu_-(s), K^\nu_+(s)\big)\Big) \geqslant \frac{\mathcal A s^4 +\mathcal Bs^2 + \mathcal C}{s^2(1 +s^2)V^2}
\end{align}
holds with
\begin{align*}
 \mathcal A &:=V^2(1 -\eta)^2,\\
 \mathcal B &:=V^2(1 -2\nu^2) -(1 +2V^2 +2\nu^2V^2 +\nu^4V^4)\eta + (1+ V^2 +\nu^4V^4)\eta^2,\\
 \mathcal C &:=\nu^4V^2 -\nu^2(1 +V^2 +\nu^2V^4 +\nu^4V^4)\eta +\nu^4V^2(1 +V^2)\eta^2.
\end{align*}
There exists $\eta_\nu >0$ such that for any $\eta\in [0, 2\eta_\nu]$ the coefficients $\mathcal A$, $\mathcal B$ and $\mathcal C$ are strictly positive, hence also the right hand side of \eqref{determinant}. Since for $\eta =0$ both eigenvalues of $\big(M_{+}^\nu(s)\big)^*M_{+}^\nu(s)$ are positive, both eigenvalues of \[\big(M_{+}^\nu(s)\big)^*M_{+}^\nu(s) -\eta\diag\big(K^\nu_-(s), K^\nu_+(s)\big)\] are non-negative for all $s\in \mathbb R\setminus\{0\}$ provided $\eta\in [0, \eta_\nu]$.
\end{proof}

\begin{rem}
 It is easy to see that
\begin{align}\label{eta_nu as an inf}
 \eta_\nu =\inf_{s \in\mathbb R\setminus\{0\}}\eta_-^\nu(s),
\end{align}
where $\eta_-^\nu(s)$ is the smallest of the two solutions $\eta$ of
\begin{align*}
\det\Big(\big(M_{+}^\nu(s)\big)^*M_{+}^\nu(s) -\eta\diag\big(K^\nu_-(s), K^\nu_+(s)\big)\Big) =0.
\end{align*}
Numerical analysis indicates that the infimum in \eqref{eta_nu as an inf} is achieved for $s =+0$ and is thus equal to
\begin{align*}
 &\frac12\Bigg(\frac{{\nu}^{2}+1}{\big(1-V_{1/2}(\mathrm{i}\beta)^{-1}\big)^2}+\nu^{2}V_{-1/2}(\mathrm{i}\beta)^2 \\
 &-\sqrt{\bigg(\frac{{\nu}^{2}+1}{\big(1-V_{1/2}(\mathrm{i}\beta)^{-1}\big)^2}+\nu^{2}V_{-1/2}(\mathrm{i}\beta)^2\bigg)^2-\frac{4\nu^{4}V_{-1/2}(\mathrm{i}\beta)^2}{\big(1-V_{1/2}(\mathrm{i}\beta)^{-1}\big)^2}}\ \Bigg).
\end{align*}
\end{rem}

The final result of this section is
\begin{lem}\label{squares comparison lemma}
The inequality
\begin{align}\label{squares comparison}
(D_{\pm1/2}^\nu)^2 \geqslant \eta_\nu\big(\mathcal U_{\pm1/2}^*\diag(H_{1/2 \mp1/2}^{(V_{\mp1/2}(\mathrm i\beta))^{-1}}, H_{1/2 \pm1/2}^{(V_{\pm1/2}(\mathrm i\beta))^{-1}})\mathcal U_{\pm1/2}\big)^2
\end{align}
holds for any $\nu\in (0, 1/2]$ with $\eta_\nu$ defined in Lemma \ref{critical Dirac via Kato bound lemma}.
\end{lem}

\begin{proof}
For arbitrary $\Psi\in\mathfrak C^\nu_{\pm1/2}$ we use \eqref{U Psi representation} to represent $\mathcal U_{\pm1/2}\Psi$.
Applying Lemmata \ref{maximal subdomain transformed lemma}, \ref{critical Dirac via Kato bound lemma}, \ref{H alpha 0 domain lemma} and \ref{H alpha m domains lemma} together with Equation \eqref{H_0 on chi} we get
\begin{align*}
 \|D_{\pm1/2}^\nu\Psi\|^2 &= \bigg\|M^\nu_{\pm}\binom{R^1\zeta +a\chi_\pm^\nu(\cdot +\mathrm i)}{R^1\upsilon +a\nu V_{\mp1/2}(\mathrm i\beta)\chi_\pm^\nu(\cdot +\mathrm i)}\bigg\|^2\\ &\geqslant \eta_\nu\left\|\begin{pmatrix}\bigg(1 - \dfrac{V_{\mp1/2}(\cdot +\mathrm i/2)}{V_{\mp1/2}(\mathrm i\beta)}\bigg)\big(R^1\zeta +a\chi_\pm^\nu(\cdot +\mathrm i)\big)\\ \bigg(1 - \dfrac{V_{\pm1/2}(\cdot +\mathrm i/2)}{V_{\pm1/2}(\mathrm i\beta)}\bigg)\big(R^1\upsilon +a\nu V_{\mp1/2}(\mathrm i\beta)\chi_\pm^\nu(\cdot +\mathrm i)\big)\end{pmatrix}\right\|^2\\
&=\eta_\nu\big\|\mathcal U_{\pm1/2}^*\diag(H_{1/2 \mp1/2}^{(V_{\mp1/2}(\mathrm i\beta))^{-1}}, H_{1/2 \pm1/2}^{(V_{\pm1/2}(\mathrm i\beta))^{-1}})\mathcal U_{\pm1/2}\Psi\big\|^2.
\end{align*}
Since $\mathfrak C^\nu_{\pm1/2}$ is an operator core for $D^\nu_{\pm1/2}$, we conclude \eqref{squares comparison}.
\end{proof}

\section{Non-critical channels estimate}\label{non-critical section}

\begin{lem}\label{non-critical channels lemma}
 For $\nu\in (0, 1/2]$ the operator inequalities
\begin{align}\label{non-critical channels inequality squared}
 (D_{\varkappa}^\nu)^2 \geqslant \Big(1 -\nu\big(3(16 +\nu^2)^{1/2} -5\nu\big)/8\Big)\big(\mathcal U_\varkappa^*R^1\mathcal U_\varkappa\big)^2
\end{align}
hold true for all $\varkappa \in(\mathbb Z +1/2)\setminus\{-1/2,\, 1/2\}$.
\end{lem}

\begin{proof}
As in Lemma \ref{squares comparison lemma}, it is enough to prove \eqref{non-critical channels inequality squared} on the functions from the operator core $\mathfrak C_{\varkappa}^\nu$ which, according to \eqref{operator core}, coincides with $\mathsf C_0^\infty(\mathbb R_+, \mathbb C^2)$.

With the help of Lemma \ref{transformed operators lemma}, \eqref{Upsi} and \eqref{D^nu}  we get for every $\varphi\in\mathsf C_0^\infty(\mathbb R_+, \mathbb C^2)$
\begin{align*}
 &\|\mathcal U_\varkappa^*R^1\mathcal U_\varkappa\varphi\|^2 =\Big\|\bigoplus_{\widetilde\varkappa\in \mathbb Z +1/2}\delta_{\widetilde\varkappa, \varkappa}R^1\mathcal U_{\varkappa}\varphi\Big\|^2\\ &=\Big\|\mathcal{ST}(-\Delta)^{1/2}(\mathcal{ST})^*\bigoplus_{\widetilde\varkappa\in \mathbb Z +1/2}\delta_{\widetilde\varkappa, \varkappa}\,\mathcal U_{\varkappa}\varphi\Big\|^2\\ &=\Big\|\mathcal{A}(-\mathrm i\boldsymbol\sigma\cdot\nabla)\mathcal A^*(\mathcal{STA}^*)^*\bigoplus_{\widetilde\varkappa\in \mathbb Z +1/2}\delta_{\widetilde\varkappa, \varkappa}\,\mathcal U_{\varkappa}\varphi\Big\|^2 =\|D_\varkappa^0\varphi\|^2.
\end{align*}
It is thus enough to prove \eqref{non-critical channels inequality squared} with $D_\varkappa^0$ instead of $\mathcal U_\varkappa^*R^1\mathcal U_\varkappa$.

For $b\in \mathbb R$ we introduce a family of matrix-functions
\begin{align*}
 A^\nu_\varkappa(b, s) :=\begin{pmatrix}\nu^2 +b\big(s^2 +(1/2 -\varkappa)^2\big)& 2\nu(\mathrm is +\varkappa)\\ 2\nu(-\mathrm is +\varkappa) & \nu^2 +b\big(s^2 +(\varkappa +1/2)^2\big)\end{pmatrix}, \quad s\in\mathbb R.
\end{align*}
A straightforward calculation using Lemma \ref{D^nu_kappa lemma}, \eqref{d_kappa} and \eqref{Mellin transform} delivers
\begin{align*}
 &\|D^\nu_\varkappa\varphi\|^2= \|\mathcal MD^\nu_\varkappa\varphi\|^2= \int\limits_{-\infty}^\infty\big\langle (R^{-1}\mathcal M\varphi)(s), A_\varkappa^\nu(1, s)(R^{-1}\mathcal M\varphi)(s)\big\rangle\,\mathrm ds.
\end{align*}
Thus
\begin{align}\label{lower bound with A}
 &\|D_\varkappa^\nu\varphi\|^2 -(1 -b)\|D_\varkappa^0\varphi\|^2 =\int_{-\infty}^\infty\big\langle (R^{-1}\mathcal M\varphi)(s), A_\varkappa^\nu(b, s)(R^{-1}\mathcal M\varphi)(s)\big\rangle\,\mathrm ds
\end{align}
holds. The eigenvalues of $A^\nu_\varkappa(b, s)$ are given by
\begin{align}\label{eigenvalues of A}
 a^\nu_{\varkappa, \pm}(b, s):= \nu^2 +b/4 +\varkappa^2b +s^2b \pm(4\varkappa^2\nu^2 +4\nu^2s^2 +\varkappa^2b^2)^{1/2}.
\end{align}
Note that $a^\nu_{\varkappa, \pm} =a^\nu_{-\varkappa, \pm}$.

We now seek $b <1$ such that the inequality $a^\nu_{\varkappa, -}(b, s)\geqslant 0$ holds for all $\varkappa\in \mathbb N_1 +1/2$ and $s\in\mathbb R$.
We claim that, for all other parameters being fixed, $a^\nu_{\varkappa, -}(b, s)$ is an increasing function of $\varkappa\in \mathbb N_1 +1/2$ provided $b \geqslant 2\nu/\sqrt{15}$ holds. Indeed, extending \eqref{eigenvalues of A} to $\varkappa \in\mathbb R$, we get
\begin{align*}
 &a^\nu_{\varkappa +1, -}(b, s) -a^\nu_{\varkappa, -}(b, s) =\int_{\varkappa}^{\varkappa +1}\frac{\partial a^\nu_{\widetilde\varkappa, -}}{\partial \widetilde\varkappa}(b, s)\,\mathrm d\widetilde\varkappa\\
&=(2\varkappa +1)b -\int_{\varkappa}^{\varkappa +1}\frac{(4\nu^2 +b^2)\widetilde\varkappa}{\sqrt{(4\nu^2 +b^2)\widetilde\varkappa^2 +4\nu^2s^2}}\,\mathrm d\widetilde\varkappa \geqslant 4b -\sqrt{4\nu^2 +b^2} \geqslant 0.
\end{align*}
Note that $a^\nu_{3/2, -}(b, s) =a^\nu_{3/2, -}(b, -s)$. For $s >0$ and 
\begin{align*}
 b\geqslant \nu\sqrt{2(\sqrt{13}/3 -1)}\qquad \Big(> 2\nu/\sqrt{15}\Big)
\end{align*}
we have
\begin{align*}
 \frac{\partial a^\nu_{3/2, -}(b, s)}{\partial s} &= 2b s -\frac{4\nu^2s}{\sqrt{9\nu^2 +4\nu^2s^2 +9b^2/4}}\\ & \geqslant 2s\big(b -2\nu^2/\sqrt{9\nu^2 +9b^2/4}\big) \geqslant 0.
\end{align*}
Thus, provided 
\begin{align*}
 b \geqslant \nu\big(3\sqrt{16 +\nu^2} -5\nu\big)/8\qquad \Big(> \nu\sqrt{2(\sqrt{13}/3 -1)} \quad\textrm{for all }\nu\in (0, 1/2]\Big)
\end{align*}
holds, for any $s\in \mathbb R$ and $\varkappa \in \mathbb N_1+ 1/2$ we have
\begin{align*}
 a^\nu_{\varkappa, -}(b, s)\geqslant a^\nu_{3/2, -}(b, 0) =\nu^2 +5b/2 -3\sqrt{\nu^2 +b^2/4} \geqslant 0.
\end{align*}
It follows now from \eqref{lower bound with A} that 
\begin{align*}
 (D_{\varkappa}^\nu)^2 \geqslant \Big(1 -\nu\big(3(16 +\nu^2)^{1/2} -5\nu\big)/8\Big)(D_{\varkappa}^0)^2
\end{align*}
(and hence \eqref{non-critical channels inequality squared}) holds for all $\nu\in (0, 1/2]$ and $\varkappa\in(\mathbb Z +1/2)\setminus\{-1/2,\, 1/2\}$.
\end{proof}

\section{Critical lower bounds}\label{clb section}

In this section we prove lower bounds analogous to the critical hydrogen inequality introduced in Theorem 2.3 of \cite{SolovejSoerensenSpitzer} and further developed in \cite{FrankHLT}.

For $\gamma\in \mathbb R$ we introduce the quadratic form
\begin{equation*}
 \mathfrak p^\gamma[f]:= \int_{\mathbb R_+}p^\gamma\big|f(p)\big|^2\,\mathrm dp
\end{equation*}
on $\mathsf L^2\big(\mathbb R_+, (1+ p^\gamma)\mathrm dp\big)$.

The next theorem will imply a lower bound for the quadratic form of the critical operator $H_m^{\alpha_m}$. Recall the definition \eqref{alpha_c} of $\alpha_m$ and Lemma \ref{WF Coulomb lemma}.

\begin{thm}\label{critical lower bound theorem}
 For any $m\in\mathbb Z$ and $\lambda\in (0, 1)$ there exists $K_{m, \lambda}> 0$ such that for all $l >0$ the inequality
\begin{equation}\label{critical lower bound}
 \mathfrak p^1 -\alpha_m\mathfrak q_m\geqslant K_{m, \lambda}l^{\lambda -1}\mathfrak p^\lambda -l^{-1}\mathfrak p^0
\end{equation}
holds on $\mathsf L^2\big(\mathbb R_+, (1+ p)\mathrm dp\big)$.
\end{thm}

\begin{proof}
 Let $m\in\mathbb Z$, $\lambda\in (3/4, 1)$ and $f\in \mathsf L^2\big(\mathbb R_+, (1+ p)\mathrm dp\big)$. Using the non-negativity of $Q_{|m|-1/2}$, the Cauchy-Bunyakovsky-Schwarz inequality and that
\begin{equation*}
 (q+l^{\lambda -1}q^{\lambda})^{-1}\leqslant q^{-1}-l^{\lambda -1}q^{\lambda -2}+l^{2(\lambda -1)}q^{2\lambda -3} \quad\text{holds for all }q, l>0
\end{equation*}
 (which follows from $(1 +z)^{-1}\leqslant 1 -z +z^2$ for all $z\geqslant 0$ by letting $z:= (lq)^{\lambda -1}$) we obtain
 \begin{equation}\begin{split}\label{with Taylor}
 \mathfrak q_m[f] &=\frac1\pi\int_{0}^{\infty}\int_{0}^{\infty}\overline{f(p)}f(q)Q_{|m|-1/2}\bigg(\frac{1}{2}\Big(\frac{p}{q}+\frac{q}{p}\Big)\bigg)\mathrm dq\,
 \mathrm{d}p\\
 &\leqslant \frac1\pi\int_{0}^{\infty}\int_{0}^{\infty}\big|f(p)\big|^{2}Q_{|m|-1/2}\bigg(\frac{1}{2}\Big(\frac{p}{q}+\frac{q}{p}\Big)\bigg)
 \bigg(\frac{p+l^{\lambda -1}p^\lambda}{q+l^{\lambda -1}{q}^\lambda}\bigg)\mathrm dq\,\mathrm dp\\
 &\leqslant \frac1\pi\int_{0}^{\infty}\int_{0}^{\infty}\big|f(p)\big|^{2}Q_{|m|-1/2}\bigg(\frac{1}{2}\Big(\frac{p}{q}+\frac{q}{p}\Big)\bigg)\\
 &\qquad\qquad\quad\times(p+l^{\lambda -1}p^\lambda)(q^{-1}-l^{\lambda -1}q^{\lambda -2}+l^{2(\lambda -1)}q^{2\lambda -3})\mathrm dq\,\mathrm dp.
 \end{split}\end{equation}
From \eqref{Q} it is easy to find the asymptotics
\begin{equation*}
 Q_{|m|-1/2}\big((x+ x^{-1})/2\big)\sim \frac{\pi^{1/2}\Gamma\big(|m| +1/2\big)}{\Gamma\big(|m| +1\big)}\begin{cases}
                                      x^{-|m|- 1/2}, & \textrm{for }x\to +\infty;\\x^{|m| +1/2}, & \textrm{for }x\to +0,
                                     \end{cases}
\end{equation*}
which implies that the function
 \begin{equation*}\begin{split}
  V_{|m| -1/2}(z)&:=\frac1\pi\int_{0}^{\infty}Q_{|m|-1/2}\big((x+ x^{-1})/2\big)x^{-\mathrm iz -1} \mathrm dx\\ 
  &=\frac{p^{\mathrm iz}}\pi\int_{0}^{\infty}Q_{|m| -1/2}\bigg(\frac{1}{2}\Big(\frac{p}{q}+\frac{q}{p}\Big)\bigg)q^{-\mathrm iz -1}\mathrm dq
 \end{split}\end{equation*}
 is well-defined and analytic in the strip $\Big\{z\in \mathbb C: \Im z\in\big(-|m| -1/2, |m| +1/2\big)\Big\}$. It also coincides there with the function defined in \eqref{V}, as can be seen by comparing Lemma \ref{WF Coulomb lemma} with Lemmata \ref{1/r with T lemma}, \ref{R and Xi commutation lemma} and Corollary \ref{V corollary} or by a calculation as in Section VI of \cite{YaouancOliverRaynal}.

We can then rewrite the right hand side of \eqref{with Taylor} obtaining
\begin{align}\begin{split}\label{terms with V}
\mathfrak q_m[f] &= V_{|m| -1/2}(0)\mathfrak p^1[f] + \Big(V_{|m| -1/2}(0)-V_{|m| -1/2}\big(\mathrm i(\lambda -1)\big)\Big)l^{\lambda -1}\mathfrak p^\lambda[f]\\ &+\Big(V_{|m| -1/2}\big(2\mathrm i(\lambda -1)\big)-
 V_{|m| -1/2}\big(\mathrm i(\lambda -1)\big)\Big)l^{2(\lambda -1)}\mathfrak p^{2\lambda -1}[f]\\ &+V_{|m| -1/2}\big(2\mathrm i(\lambda -1)\big)l^{3(\lambda -1)}\mathfrak p^{3\lambda -2}[f].
\end{split}\end{align}
Lemmata \ref{critical constants lemma} and \ref{V monotonicity lemma} imply
 \begin{equation}\begin{split}\label{signs of V}
        V_{|m| -1/2}(0)&=\alpha_m^{-1}; \\
        V_{|m| -1/2}(0)-V_{|m| -1/2}\big(\mathrm i(\lambda -1)\big) &< 0; \\
        V_{|m| -1/2}\big(2\mathrm i(\lambda -1)\big)- V_{|m| -1/2}\big(\mathrm i(\lambda -1)\big)&\geqslant 0;\\
        V_{|m| -1/2}\big(2\mathrm i(\lambda -1)\big) &\geqslant 0.
 \end{split}\end{equation}
For every $\lambda\in(3/4, 1)$ and $\varepsilon_{1} ,\varepsilon_{2} >0$ there exist $C_{1}, C_{2}> 0$ such that the inequalities
\begin{equation}\label{powers interpolation}
  l^{2(\lambda -1)}p^{2\lambda -1}\leqslant \varepsilon_1p^\lambda l^{\lambda -1}+ C_1l^{-1}, \quad l^{3(\lambda -1)}p^{3\lambda -2}\leqslant \varepsilon_2p^\lambda l^{\lambda -1}+ C_2l^{-1}
\end{equation}
 hold for all $p, l>0$.
Substituting \eqref{powers interpolation} into \eqref{terms with V} and taking \eqref{signs of V} into account by choosing $\varepsilon_{1}$ and $\varepsilon_{2} >0$ small enough we obtain
\begin{equation}\label{q_m bound}
 \begin{split}
   \mathfrak q_m[f] \leqslant \mathfrak p^1[f]/\alpha_m -C_1(m, \lambda)l^{\lambda -1}\mathfrak p^\lambda[f] +C_2(m, \lambda)l^{-1}\mathfrak p^0[f]
 \end{split}
\end{equation}
with $C_1(m, \lambda)$, $C_2(m, \lambda) >0$ for $\lambda\in (3/4, 1)$. For $\lambda\in (0, 3/4]$, $\lambda'\in (3/4, 1)$ we can find a constant $C_3(\lambda, \lambda')> 0$ with
\begin{equation*}
 l^{\lambda'}\mathfrak p^{\lambda'}\geqslant -C_3(\lambda, \lambda')\mathfrak p^0 + l^{\lambda}\mathfrak p^{\lambda}
\end{equation*}
and get \eqref{q_m bound} for $\lambda$ from  \eqref{q_m bound} for $\lambda'$.
Rescaling $l$ and using $\alpha_m >0$ we arrive at \eqref{critical lower bound}.
 \end{proof}

\begin{cor}\label{critical lower bound corollary}
 For $m\in\mathbb Z$ and $\lambda\in (0, 1)$ the inequality
 \begin{equation}\label{critical lower bound with operators}
  H_m^{\alpha_m}\geqslant K_{m, \lambda}l^{\lambda -1}R^\lambda -l^{-1}
 \end{equation}
 holds for all $l >0$ with $K_{m, \lambda}$ as in \eqref{critical lower bound}.
\end{cor}

\begin{proof}
 For any $\varphi\in \mathfrak D(\widetilde H_m^{\alpha_m})$ we have
\begin{equation}\label{form of H_m tilde}
 \langle \varphi, \widetilde H_m^{\alpha_m}\varphi\rangle = \langle \varphi, R^1\varphi\rangle - \alpha_m\langle \varphi, V_{|m|- 1/2}(\cdot +\mathrm i/2)R^1\varphi\rangle.
\end{equation}
By \eqref{complex shifted}, the first term on the right hand side of \eqref{form of H_m tilde} coincides with $\mathfrak p^1[\mathcal M^*\varphi]$. Letting
\begin{equation*}
 \Phi:= \mathcal T^*\bigoplus_{n\in\mathbb Z}\delta_{n, m}\varphi
\end{equation*}
and using Lemma \ref{1/r with T lemma}, Corollary \ref{V corollary} and Lemma \ref{WF Coulomb lemma} we obtain
\begin{equation*}
 \langle \varphi, V_{|m|- 1/2}(\cdot +\mathrm i/2)R^1\varphi\rangle =\langle \Phi, r^{-1}\Phi\rangle =\mathfrak q_m[\mathcal M^*\varphi].
\end{equation*}
Thus \eqref{form of H_m tilde} can be written as
\begin{equation*}
 \langle \varphi, \widetilde H_m^{\alpha_m}\varphi\rangle = \mathfrak p^1[\mathcal M^*\varphi] - \alpha_m\mathfrak q_m[\mathcal M^*\varphi]
\end{equation*}
for any $\varphi\in \mathfrak D(\widetilde H_m^{\alpha_m})$. Using Theorem \ref{critical lower bound theorem}, \eqref{complex shifted} and that $H_m^{\alpha_m}$ is the Friedrichs extension of $\widetilde H_m^{\alpha_m}$ we conclude \eqref{critical lower bound with operators}.
\end{proof}

\section{Proofs of the main theorems}

\paragraph{Proof of Theorem \ref{t:abs_value_bound}.}\label{main proofs section}

\emph{1.} By Lemma \ref{critical constants lemma}
\begin{equation*}
\begin{split}
\langle\varphi, \widetilde H^\alpha_m\varphi\rangle \geqslant (1 -\alpha/\alpha_m)\langle\varphi, R^1\varphi\rangle
\end{split}
\end{equation*}
holds for all $m\in \mathbb Z$, $\alpha\in [0, \alpha_m)$ and $\varphi\in\mathfrak D(\widetilde H_m^{\alpha})$. Passing to the Friedrichs extension and using \eqref{alpha_c} we obtain
\begin{align}\label{H alpha lower bound}
 H^\alpha_m\geqslant \big(1 -\alpha V_{|m| -1/2}(0)\big)R^1.
\end{align}
By the operator monotonicity of the square root, Lemma \ref{squares comparison lemma} implies
\begin{align}\label{moduli comparison}
|D_{\pm1/2}^\nu| \geqslant \eta_\nu^{1/2}\mathcal U_{\pm1/2}^*\diag\Big(H_{1/2 \mp1/2}^{(V_{\mp1/2}(\mathrm i\beta))^{-1}}, H_{1/2 \pm1/2}^{(V_{\pm1/2}(\mathrm i\beta))^{-1}}\Big)\mathcal U_{\pm1/2}.
\end{align}
With \eqref{H alpha lower bound} we conclude
\begin{align}\label{critical channels}
 |D_{\pm1/2}^\nu| \geqslant \eta_\nu^{1/2}\min\bigg\{1 -\frac{V_{-1/2}(0)}{V_{-1/2}(\mathrm i\beta)},\, 1 -\frac{V_{1/2}(0)}{V_{1/2}(\mathrm i\beta)}\bigg\}\mathcal U_{\pm1/2}^*R^1\mathcal U_{\pm1/2}.
\end{align}
Lemma \ref{non-critical channels lemma} implies, in its turn, the estimate
\begin{align}\label{non-critical moduli}
 |D_{\varkappa}^\nu| \geqslant \Big(1 -\nu\big(3(16 +\nu^2)^{1/2} -5\nu\big)/8\Big)^{1/2}\,\mathcal U_\varkappa^*R^1\mathcal U_\varkappa.
\end{align}
Combining it with \eqref{critical channels}, \eqref{D^nu} and Lemma \ref{transformed operators lemma} we arrive at
\begin{align}\label{Upsi inbetween}
 |D^\nu| &\geqslant C_\nu\mathcal A^*\bigg(\bigoplus_{\varkappa\in \mathbb Z +1/2}\mathcal U_{\varkappa}^*R^1\mathcal U_{\varkappa}\bigg)\mathcal A =C_\nu\mathcal T^*\bigg(\bigoplus_{m\in \mathbb Z}R^1\bigg)\mathcal T =C_\nu\sqrt{-\Delta}
\end{align}
with
\begin{align*}
 C_\nu :=\min\bigg\{\eta_\nu^{1/2}\bigg(1 -\frac{V_{-1/2}(0)}{V_{-1/2}(\mathrm i\beta)}\bigg),&\, \eta_\nu^{1/2}\bigg(1 -\frac{V_{1/2}(0)}{V_{1/2}(\mathrm i\beta)}\bigg),\\ &\Big(1 -\nu\big(3(16 +\nu^2)^{1/2} -5\nu\big)/8\Big)^{1/2}\bigg\}.
\end{align*}

\emph{2.} Corollary \ref{critical lower bound corollary} and \eqref{moduli comparison} imply
\begin{align}\label{critical critical}
 |D_{\pm1/2}^\nu| \geqslant \eta_\nu^{1/2}\big(\min\{K_{0, \lambda}, K_{1, \lambda}\}l^{\lambda -1}\mathcal U_{\pm1/2}^*R^\lambda\mathcal U_{\pm1/2} -l^{-1}\big).
\end{align}
For $\varkappa\in(\mathbb Z +1/2)\setminus\{-1/2, 1/2\}$ we combine \eqref{non-critical moduli} and the simple inequality
\begin{align*}
 R^1 \geqslant \lambda^{-\lambda}(1 -\lambda)^{\lambda -1} l^{\lambda -1}R^\lambda -l^{-1}
\end{align*}
which follows from the spectral theorem. This together with \eqref{critical critical} implies \eqref{l bound} with
\begin{align*}
 K_\lambda :=\min\big\{\eta_{1/2}^{\lambda/2}K_{0,\lambda},\, \eta_{1/2}^{\lambda/2}K_{1, \lambda}, \lambda^{-\lambda}(1 -\lambda)^{\lambda -1} 2^{-5\lambda/2}(37-3\sqrt{65})^{\lambda/2}\big\}
\end{align*}
by a calculation analogous to \eqref{Upsi inbetween}.

\paragraph{Proof of Corollary \ref{defining corollary}.}
Under the assumptions of Corollary \ref{defining corollary} for any $\varepsilon >0$ there exists a decomposition
\begin{align}\label{V and B epsilon}
 V =V_\varepsilon + B_\varepsilon
\end{align}
with
\begin{align*}
\|\tr V_\varepsilon^{2 +\gamma}\|_{\mathsf L^1(\mathbb R^2)}< \varepsilon^{2 +\gamma} \quad \textrm{ and }\quad B_\varepsilon\in \mathsf L^\infty(\mathbb R^2, \mathbb C^{2\times 2}).
\end{align*}
By H\"older and Sobolev inequalities there exists $C_S >0$ such that for any $\varphi\in P_+^\nu\mathfrak D\big(|D^\nu|^{1/2}\big)$ we get
\begin{align}\label{Hoelder and Sobolev}
 \bigg|\int\limits_{\mathbb{R}^2}\big\langle \varphi(\mathbf x), V_\varepsilon(\mathbf x)\varphi (\mathbf x)\big\rangle \mathrm{d}\mathbf{x}\bigg| \leqslant\varepsilon\|\varphi\|^2_{\mathsf L^{\frac{4 +2\gamma}{1 +\gamma}}(\mathbb R^2)} \leqslant \varepsilon C_S\big\|(-\Delta)^{1/(4 +2\gamma)}\varphi\big\|^2.
\end{align}
Now \eqref{CLR reduction} and the estimate $(-\Delta)^{1/(2 +\gamma)} \leqslant (-\Delta)^{1/2} +1$ imply 
\begin{align}\label{case nu <1/2}
 \big\|(-\Delta)^{1/(4 +2\gamma)}\varphi\big\|^2 \leqslant C_\nu^{-1}\big\||D^\nu|^{1/2}\varphi\big\|^2 +\|\varphi\|^2,
\end{align}
for any $\nu\in [0, 1/2)$, $\gamma\geqslant 0$.
For $\nu =1/2$ and $\gamma >0$ we use \eqref{l bound} with $\lambda:= 2/(2 +\gamma)$, $l:= K_{2/(2 +\gamma)}^{(2 +\gamma)/\gamma}$ obtaining
\begin{align}\label{case nu =1/2}
 \big\|(-\Delta)^{1/(4 +2\gamma)}\varphi\big\|^2 \leqslant \big\||D^{1/2}|^{1/2}\varphi\big\|^2 +K_{2/(2 +\gamma)}^{-(2 +\gamma)/\gamma}\|\varphi\|^2.
\end{align}
Combining \eqref{V and B epsilon} and \eqref{Hoelder and Sobolev} with \eqref{case nu <1/2} or \eqref{case nu =1/2} we conclude that $V$ is an infinitesimal form perturbation of $\mathfrak d^\nu(0, 0)$ for all $(\nu, \gamma)\in \big([0, 1/2]\times [0, \infty)\big)\setminus \big\{(1/2, 0)\big\}$. This together with \eqref{w} implies that $\mathfrak d^\nu(\mathfrak w, V)$ is bounded from below by some $-M\in \mathbb R$ and that
\[\mathfrak d^\nu(\mathfrak w, V)[\cdot] +(M +1)\|\cdot\|^2 \quad\textrm{and} \quad\mathfrak d^\nu(0, 0)[\cdot] +\|\cdot\|^2\] are equivalent norms on $P_+^\nu\mathfrak D\big(|D^\nu|^{1/2}\big)$ (see e.g. the proof of Theorem X.17 in \cite{ReedSimonII}).

\paragraph{Proof of Theorem \ref{CLR theorem}.} Using the spectral theorem and \eqref{CLR reduction} we obtain
\begin{align*}
 &\rank \big(D^\nu(\mathfrak w, V)\big)_- =\sup \dim\Big\{\mathcal X \textrm{ subspace of }P_+^\nu\mathfrak D\big(|D^\nu|^{1/2}\big):\\ &\hspace{5cm} \mathfrak d^\nu(\mathfrak w, V)[\psi]< 0 \textrm{ for all }\psi\in \mathcal X\setminus\{0\}\Big\}\\ &\leqslant \sup \dim\bigg\{\mathcal X \textrm{ subspace of }\mathsf H^{1/2}(\mathbb R^2, \mathbb C^2):\textrm{ for all }\psi\in \mathcal X\setminus\{0\}\\& \qquad\qquad\qquad \big\|(-\Delta)^{1/4}\psi\big\|^2 -C_\nu^{-1}\int_{\mathbb{R}^2}\big\langle \psi(\mathbf x), V(\mathbf x)\psi (\mathbf x)\big\rangle \mathrm{d}\mathbf{x}< 0 \textrm{ holds.}\bigg\}\\ &= \rank\big((-\Delta)^{1/2} - C_\nu^{-1}V\big)_-,
\end{align*}
where the operator on the right hand side is the one considered in Example 3.3 of \cite{Frank2014CLR}.
The statement now follows from \eqref{CLR for fractional Laplacian} with
\begin{align*}
 C^{\mathrm{CLR}}_\nu :=4C_\nu^{-2}/\pi.
\end{align*}

\paragraph{Proof of Theorem \ref{LT theorem}.}
For $\nu <1/2$, the statement follows from Theorem \ref{CLR theorem} in the usual way. First, we pass to the integral representation
\begin{align}\begin{split}\label{integral formula}
  \tr\big(D^\nu(\mathfrak w, V)\big)^\gamma_- &=\gamma\int_0^\infty  \rank \big(D^\nu(\mathfrak w, V) +\tau\big)_-\tau^{\gamma -1}\mathrm d\tau\\& \leqslant \gamma\int_0^\infty  \rank \Big(D^\nu\big(\mathfrak w, (V -\tau)_+\big)\Big)_-\tau^{\gamma -1}\mathrm d\tau.
\end{split}\end{align}
Now, applying \eqref{CLR}, we can estimate the right hand side of \eqref{integral formula} by
\begin{align}\label{CLR integral}
 \gamma C^{\mathrm{CLR}}_\nu\int_{\mathbb R^2}\int_0^\infty \tr \big(V(\mathbf x) -\tau\big)_+^2\tau^{\gamma -1}\mathrm d\tau\dd\mathbf x.
\end{align}

For $\mathbf x\in\mathbb R$ let $v_{1,2}(\mathbf x)$ be the eigenvalues of $V(\mathbf x)$. Computing the trace in the eigenbasis of $V(\mathbf x)$ we obtain for all $\tau \geqslant 0$
\begin{align}\label{eigenbasis}
 \tr \big(V(\mathbf x) -\tau\big)_+^2= \sum_{j =1}^2\big(v_j(\mathbf x) -\tau\big)_+^2.
\end{align}
Substituting \eqref{eigenbasis} into \eqref{CLR integral} and computing the integrals we derive \eqref{LT} with
\begin{align*}
 C^{\mathrm{LT}}_{\nu, \gamma} =\frac{2C_{\nu}^{\mathrm{CLR}}}{(\gamma +1)(\gamma +2)}, \quad\text{for } \nu <1/2.
\end{align*}

For $\nu =1/2$, the inequality \eqref{LT} follows from \eqref{l bound} by a calculation similar to the one in the proof of Theorem 1.1 in \cite{FrankHLT}. Namely, proceeding analogously to the proof of Theorem \ref{CLR theorem}, but using \eqref{l bound} instead of \eqref{CLR reduction}, we observe the inequalities
\begin{align}\begin{split}\label{rank with tau}
 &\rank \big(D^{1/2}(\mathfrak w, V) +\tau\big)_- \leqslant \rank\Big((-\Delta)^{\lambda/2} -K_\lambda^{-1}l^{1 -\lambda}\big(V +(l^{-1} -\tau)\big)\Big)_-
\end{split}\end{align}
for all $\lambda\in (0, 1)$, $\tau, l >0$. We now let $l:= (\sigma\tau)^{-1}$ with $\sigma\in (0, 1)$ and estimate the right hand side of \eqref{rank with tau} from above with the help of \eqref{CLR for fractional Laplacian} by
\begin{align*}
 (2\pi \lambda)^{-1}(1 -\lambda/2)^{1 -4/\lambda} K_\lambda^{-2/\lambda}(\sigma\tau)^{2(\lambda -1)/\lambda}\int_{\mathbb R^2} \tr \big(V(\mathbf x) -(1 -\sigma)\tau\big)_+^{2/\lambda}\dd\mathbf x.
\end{align*}
Substituting this into \eqref{integral formula} and integrating in $\tau$ we get for $2/(2+ \gamma) <\lambda <1$
\begin{align*}
 \tr\big(D^\nu(\mathfrak w, V)\big)^\gamma_- \leqslant C^{\mathrm{LT}}_{1/2, \gamma}(\lambda, \sigma)\int_{\mathbb R^2} \tr \big(V(\mathbf x)\big)^{2+ \gamma}\dd\mathbf x
\end{align*}
with
\begin{align*}
 C^{\mathrm{LT}}_{1/2, \gamma}(\lambda, \sigma) :=\gamma\Big(1 -\frac\lambda2\Big)^{1 -\frac4\lambda}\frac{\Gamma\big(2 +\gamma -\frac2\lambda\big)\Gamma\big(1 +\frac2\lambda\big)}{2\pi \lambda K_\lambda^{\frac2\lambda}\Gamma(3 + \gamma)}\sigma^{2 -\frac2\lambda}(1 -\sigma)^{-\gamma -2 +\frac2\lambda}.
\end{align*}
The estimate \eqref{LT} follows with
\begin{equation*}
 C^{\mathrm{LT}}_{1/2, \gamma} :=\min_{\substack{\lambda\in (2/(2+ \gamma), 1)\\ \sigma\in (0, 1)}}C^{\mathrm{LT}}_{1/2, \gamma}(\lambda, \sigma) =\min_{\lambda\in (2/(2+ \gamma), 1)}C^{\mathrm{LT}}_{1/2, \gamma}\Big(\lambda, \frac{2(1 -\lambda)}{\lambda\gamma}\Big).
\end{equation*}


\begin{thebibliography}{10}

\bibitem{dlmf}
{NIST} digital library of mathematical functions.
\newblock http://dlmf.nist.gov/, Release 1.0.10 of 2015-08-07.
\newblock Online companion to \cite{Olver:2010:NHMF}.

\bibitem{BalinskyEvans}
A.~A. Balinsky and W.~D. Evans.
\newblock {\em Spectral analysis of relativistic operators}.
\newblock Imperial {C}ollege {P}ress, {L}ondon, 2011.

\bibitem{BirmanSolomjak}
M.~S. Birman and M.~Z. Solomjak.
\newblock {\em Spectral theory of selfadjoint operators in {H}ilbert space}.
\newblock Mathematics and its Applications (Soviet Series). D. Reidel
  Publishing Co., Dordrecht, 1987.

\bibitem{Bouzouina}
A.~Bouzouina.
\newblock Stability of the two-dimensional {B}rown-{R}avenhall operator.
\newblock {\em Proc. Roy. Soc. Edinburgh Sect. A}, 132(5):1133--1144, 2002.

\bibitem{CastroNeto_et_al}
A.~H. Castro~Neto, F.~Guinea, N.~M.~R. Peres, K.~S. Novoselov, and A.~K. Geim.
\newblock The electronic properties of graphene.
\newblock {\em Rev. Mod. Phys.}, 81:109--162, 2009.

\bibitem{Cwikel}
M.~Cwikel.
\newblock Weak type estimates for singular values and the number of bound
  states of {S}chr\"odinger operators.
\newblock {\em Ann. of Math. (2)}, 106(1):93--100, 1977.

\bibitem{Downing2011}
C.~A. Downing, D.~A. Stone, and M.~E. Portnoi.
\newblock Zero-energy states in graphene quantum dots and rings.
\newblock {\em Phys. Rev. B}, 84:155437, Oct 2011.

\bibitem{EggerSiedentop2010}
R.~Egger, A.~D. Martino, H.~Siedentop, and E.~Stockmeyer.
\newblock {Multiparticle equations for interacting Dirac fermions in
  magnetically confined graphene quantum dots}.
\newblock {\em Journal of Physics A: Mathematical and Theoretical},
  43(21):215202, 2010.

\bibitem{EkholmFrank}
T.~Ekholm and R.~L. Frank.
\newblock On {L}ieb-{T}hirring inequalities for {S}chr\"odinger operators with
  virtual level.
\newblock {\em Comm. Math. Phys.}, 264(3):725--740, 2006.

\bibitem{EstebanLoss}
M.~J. Esteban and M.~Loss.
\newblock Self-adjointness via partial {H}ardy-like inequalities.
\newblock In {\em Mathematical results in quantum mechanics}, pages 41--47.
  World Sci. Publ., Hackensack, NJ, 2008.

\bibitem{FrankHLT}
R.~L. Frank.
\newblock A simple proof of {H}ardy-{L}ieb-{T}hirring inequalities.
\newblock {\em Comm. Math. Phys.}, 290(2):789--800, 2009.

\bibitem{Frank2014CLR}
R.~L. Frank.
\newblock Cwikel's theorem and the {CLR} inequality.
\newblock {\em J. Spectr. Theory}, 4(1):1--21, 2014.

\bibitem{FrankLiebSeiringer}
R.~L. Frank, E.~H. Lieb, and R.~Seiringer.
\newblock Hardy-{L}ieb-{T}hirring inequalities for fractional {S}chr\"odinger
  operators.
\newblock {\em J. Amer. Math. Soc.}, 21(4):925--950, 2008.

\bibitem{Friedrichs}
K.~Friedrichs.
\newblock {Spektraltheorie halbbeschr{\"a}nkter Operatoren und Anwendung auf
  die Spektralzerlegung von Differentialoperatoren}.
\newblock {\em Mathematische Annalen}, 109(1):465--487, 1934.

\bibitem{GutierrezBrown2016}
C.~Gutierrez, L.~Brown, C.-J. Kim, J.~Park, and A.~N. Pasupathy.
\newblock {Klein tunnelling and electron trapping in nanometre-scale graphene
  quantum dots}.
\newblock {\em Nat Phys}, advance online publication, 2016.

\bibitem{LaptevWeidl}
A.~Laptev and T.~Weidl.
\newblock Recent results on {L}ieb-{T}hirring inequalities.
\newblock In {\em Journ\'ees ``\'{E}quations aux {D}\'eriv\'ees {P}artielles''
  ({L}a {C}hapelle sur {E}rdre, 2000)}, pages Exp.\ No.\ XX, 14. Univ. Nantes,
  Nantes, 2000.

\bibitem{YaouancOliverRaynal}
A.~{Le Yaouanc}, L.~Oliver, and J.-C. Raynal.
\newblock {The Hamiltonian $(p^2+m^2)^{1/2}-\alpha/r$ near the critical value
  $\alpha_c=2/\pi$}.
\newblock {\em Journal of Mathematical Physics}, 38(8):3997--4012, 1997.

\bibitem{LeeWong2016}
J.~Lee, D.~Wong, J.~J. Velasco, J.~F. Rodriguez-Nieva, S.~Kahn, H.-Z. Tsai,
  T.~Taniguchi, K.~Watanabe, A.~Zettl, F.~Wang, L.~S. Levitov, and M.~F.
  Crommie.
\newblock {Imaging electrostatically confined Dirac fermions in graphene
  quantum dots}.
\newblock {\em Nat Phys}, advance online publication, 2016.

\bibitem{Lieb1976}
E.~Lieb.
\newblock Bounds on the eigenvalues of the {L}aplace and {S}chr{\"o}dinger
  operators.
\newblock {\em Bull. Amer. Math. Soc.}, 82(5):751--753, 1976.

\bibitem{LiebThirring}
E.~H. Lieb and W.~E. Thirring.
\newblock {\em Studies in Mathematical Physics: Essays in Honor of {V}alentine
  {B}argmann}, chapter Inequalities for the moments of the eigenvalues of the
  {S}chr{\"o}dinger Hamiltonian and their relation to {S}obolev inequalities,
  pages 269--304.
\newblock Princeton University Press, Princeton, {N}ew {J}ersey, 1976.

\bibitem{Virtual_level}
S.~Morozov and D.~M{\"u}ller.
\newblock On the virtual levels of positively projected massless
  {C}oulomb-{D}irac operators.
\newblock {\em arXiv preprint}, 1607.08902, 2016.

\bibitem{Mueller}
D.~M{\"u}ller.
\newblock {Minimax principles, {H}ardy-{D}irac inequalities and operator cores
  for two and three dimensional {C}oulomb-{D}irac operators}.
\newblock {\em Documenta Mathematica}, 21:1151--1169, 2016.

\bibitem{Olver:2010:NHMF}
F.~W.~J. Olver, D.~W. Lozier, R.~F. Boisvert, and C.~W. Clark, editors.
\newblock {\em NIST Handbook of Mathematical Functions}.
\newblock Cambridge University Press, New York, NY, 2010.
\newblock Print companion to \cite{dlmf}.

\bibitem{Pereira-Nilsson-Castro_Neto}
V.~M. Pereira, J.~Nilsson, and A.~H. Castro~Neto.
\newblock Coulomb impurity problem in graphene.
\newblock {\em Phys. Rev. Lett.}, 99:166802, 2007.

\bibitem{ReedSimonII}
M.~Reed and B.~Simon.
\newblock {\em Methods of modern mathematical physics. {II}. {F}ourier
  analysis, self-adjointness}.
\newblock Academic Press [Harcourt Brace Jovanovich, Publishers], New
  York-London, 1975.

\bibitem{Rozenblum}
G.~V. Rozenbljum.
\newblock Distribution of the discrete spectrum of singular differential
  operators.
\newblock {\em Dokl. Akad. Nauk SSSR}, 202:1012--1015, 1972.

\bibitem{SilvestrovEfetov2007}
P.~G. Silvestrov and K.~B. Efetov.
\newblock Quantum dots in graphene.
\newblock {\em Phys. Rev. Lett.}, 98:016802, 2007.

\bibitem{SolovejSoerensenSpitzer}
J.~P. Solovej, T.~{\O}. S{\o}rensen, and W.~L. Spitzer.
\newblock Relativistic {S}cott correction for atoms and molecules.
\newblock {\em Comm. Pure Appl. Math.}, 63(1):39--118, 2010.

\bibitem{Thaller}
B.~Thaller.
\newblock {\em The {D}irac equation}.
\newblock Texts and Monographs in Physics. Springer-Verlag, Berlin, 1992.

\bibitem{Titchmarsh}
E.~C. Titchmarsh.
\newblock {\em Introduction to the theory of {F}ourier integrals}.
\newblock Oxford university press, 1948.

\bibitem{Wallace1947}
P.~R. Wallace.
\newblock The band theory of graphite.
\newblock {\em Phys. Rev.}, 71:622--634, May 1947.

\bibitem{Warmt}
C.~Warmt.
\newblock {\em Semiklassische {A}symptotik der {R}esolvente eines
  {D}iracoperators}.
\newblock Dissertation, {L}{M}{U} {M}unich, 2011.

\bibitem{Weidmann1971}
J.~Weidmann.
\newblock Oszillationsmethoden f\"ur {S}ysteme gew\"ohnlicher
  {D}ifferentialgleichungen.
\newblock {\em Math. Z.}, 119:349--373, 1971.

\bibitem{Weidmann1}
J.~Weidmann.
\newblock {\em Linear operators in {H}ilbert spaces}, volume~68 of {\em
  Graduate Texts in Mathematics}.
\newblock Springer-Verlag, New York-Berlin, 1980.

\bibitem{WhittakerWatson}
E.~T. Whittaker and G.~N. Watson.
\newblock {\em A course of modern analysis}.
\newblock Cambridge Mathematical Library. Cambridge University Press,
  {C}ambridge, 1996.
\newblock Reprint of the fourth (1927) edition.

\end{thebibliography}

\end{document}